\documentclass[journal]{IEEEtran}
\usepackage{textcomp}
\usepackage{xcolor}
\usepackage{amsmath,amsfonts}
\usepackage{algorithmic}
\usepackage{algorithm}
\usepackage{array}
\usepackage{amsthm}
\usepackage{amssymb}
\usepackage[caption=false,font=footnotesize,labelfont=rm,textfont=sf]{subfig}
\usepackage{colortab}
\usepackage{colortbl}
\usepackage{booktabs} 
\usepackage{threeparttable}
\usepackage{mathrsfs}
\usepackage{textcomp}
\usepackage{stfloats}
\usepackage{url}
\usepackage{verbatim}
\usepackage{graphicx}
\usepackage{setspace}
\usepackage{cite}
\usepackage{enumitem}
\usepackage{bm}
\begin{document}
\title{Unified framework for outage-constrained rate maximization in secure ISAC under various sensing metrics}
\author{Hancheng~Zhu, Zongze~Li, Yik-Chung~Wu
	\thanks{Hancheng Zhu and Yik-Chung Wu are with the Department of Electrical and Electronic Engineering, The University of Hong Kong, Hong Kong (email: u3006551@hku.hk, ycwu@eee.hku.hk). \emph{Corresponding author: Zongze Li, Yik-Chung~Wu}.}
	\thanks{Zongze Li is with the Peng Cheng Laboratory, Shenzhen 518038, China (e-mail: lizz@pcl.ac.cn).}}
\maketitle

\newtheorem{myLem}{\textbf{Lemma}}
\newtheorem{myPro}{\textbf{Proposition}}
\newtheorem{myCol}{\textbf{Corollary}}
\newtheorem{myExa}{\textbf{Example}}
\newtheorem{myRem}{\textbf{Remark}}
\newtheorem{myAss}{\textbf{Assumption}}
\newtheorem{myDef}{\textbf{Definition}}
\newtheorem{mySta}{\textbf{Statement}}
\newtheorem{myThe}{\textbf{Theorem}}
\newtheorem{myProp}{\textbf{Property}}

\begin{abstract}
	Integrated sensing and communication (ISAC) is poised to redefine the landscape of wireless networks by seamlessly combining data transmission and environmental sensing. However, ISAC systems remain susceptible to eavesdropping, especially under uncertainty in eavesdroppers' channel state information, which can lead to secrecy outages. On the other hand, diverse and complex sensing performance requirements further complicate resource optimization, often requiring custom solutions for each scenario. To this end, this paper introduces a unified optimization framework that holistically addresses both the worst-case user secrecy rate and the sum secrecy rate across multiple users. Besides putting the two commonly used objectives into a single but flexible objective function, the framework accurately controls secrecy outage probabilities while accommodating a broad spectrum of sensing constraints. To solve such a general problem, we integrate the sensing requirements into the objective function through an auxiliary variable. This enables efficient alternating optimization and the proposed approach is theoretically guaranteed to converge to at least a stationary point of the original problem. Extensive simulation results show that the proposed framework consistently achieves higher optimized secrecy rates under various sensing constraints compared to existing methods. These results underscore the proposed unified framework's superiority and versatility in secure ISAC systems.
\end{abstract}

\section{Introduction}
In the development of 6G networks, emerging applications such as autonomous driving, smart cities and virtual reality require high-precision sensing as well as high-quality communications \cite{Tan:18}. As sensing systems and wireless communications share many similarities, the traditional way of separately designing these two systems may not be efficient enough in the future \cite{Bliss:14}. To address this challenge, integrated sensing and communication (ISAC) grows to be a promising technology due to its potential of saving resources from both systems without scarifying performance \cite{Shihang:24}.

In designing ISAC systems, a key is on balancing the performance of communications and sensing. In general, we can optimize the communication quality under sensing requirement constraints or vice versa. However, as in conventional communication only systems, due to the open nature of wireless channels, the transmitted information from ISAC systems could be easily exposed to eavesdroppers, making the system vulnerable to unauthorized access and data breaches \cite{Zhijun:24}. Therefore, secure ISAC has received increasing attention recently \cite{Nanchi:23,Wai-Yiu:23}. 

In particular, to mitigate the adverse effects of eavesdroppers in secure ISAC, the most commonly employed strategy is through beamforming design \cite{Hou:24,Yuemin:24,Zhutian:22}. Although these results are promising, the assumption of perfect channel state information (CSI) of the eavesdroppers' channels is often overly optimistic and impractical in real-world scenarios. Consequently, most existing studies focus on addressing the physical layer security (PLS) problem in ISAC systems with estimated eavesdropper channels. For example, \cite{Xingyu:24,Hanbo:24,Yuchen:25} primarily aim to optimize beamforming strategies to reduce the detection probability of one or multiple eavesdroppers, thereby enhancing the security and reliability of the communication system. As a step further, techniques leveraging artificial noise \cite{SuNanchi:21,Peng:22} and non-orthogonal multiple access (NOMA) \cite{Yinhong:24,Ziwei:25} have been proposed to improve secrecy performance, particularly in environments characterized by jamming or scattering.
	
Recently, emerging elements such as unmanned aerial vehicles (UAVs) \cite{Yue:25} and intelligent reflective surfaces (IRS) \cite{Chen:25,Meng:24} have also been integrated into secure ISAC systems to counteract threats posed by multiple eavesdroppers. In these works, the eavesdropper channels are modeled with bounded errors \cite{SuNanchi:21,Xingyu:24,Hanbo:24,Yuchen:25,Yinhong:24,Ziwei:25,Yue:25,Chen:25,Meng:24}, or Gaussian error distributions \cite{Peng:22,Yuchen:25} around the true channels. For the bounded error case, \cite{Hanbo:24} considers the average channel uncertainties through their sample mean. Moreover, robust system designs that account for the worst eavesdroppers' channel uncertainties are considered and solved either through linear matrix inequalities (LMIs) \cite{Xingyu:24,Yuchen:25}, S-procedure \cite{SuNanchi:21,Yinhong:24,Ziwei:25,Chen:25,Meng:24} or quadratic reformulation \cite{Yue:25}. When Gaussian errors are assumed, conservative methods such as Bernstein-type inequalities (BTI) \cite{Peng:22,Yuchen:25} are employed to ensure security constraints are satisfied. Although these estimated CSI assumptions are more realistic than the idealized perfect CSI, obtaining even the estimated channels of passive and covert eavesdroppers remains infeasible in many practical situations \cite{Wei:17}. This challenge has been recognized in PLS in other wireless systems \cite{Xi:15,Qi:14,Hui-Ming:15,Haitao:23}, yet corresponding research in the context of secure ISAC systems remains limited.

On the other hand, in a multi-user setting, balancing performance among users is a critical challenge. Common objectives include worst-user rate maximization \cite{Wenyi:24,Salem:2025,Xipeng:24,Xuehua:25,Hou:24} and weighted sum rate maximization \cite{Zechen:24,Ruiwei:23,Tejaswini:24,Yuemin:24}. For example, \cite{Wenyi:24,Zechen:24} leverage the mobility of UAVs to mitigate eavesdroppers' risk via strategically planning UAV trajectories and beamforming at the base station (BS), thereby enhancing user transmission rates. Additionally, \cite{Xuehua:25} employs a two-stage transmission protocol: initial environmental sensing followed by data transmission, to reduce the impact of eavesdroppers. Although both worst-user rate maximization and weighted sum rate maximization aim to optimize data transmission, existing studies typically treat these as independent problems. Notably, worst-user rate maximization is generally considered to be more challenging than sum rate maximization. This complexity arises from the discrete nature of selection variables in worst-user rate maximization problems. Current methods address this by employing deep reinforcement learning \cite{Xipeng:24}, or replacing the discrete selection with auxiliary variables \cite{Wenyi:24,Salem:2025,Xuehua:25,Hou:24} and then convert it into a maximization framework. In contrast, sum rate maximization does not have such discrete handling steps, resulting in markedly different algorithmic approaches.

Apart from the differences in the objective function, sensing constraints in ISAC systems can also be markedly distinct. For instance, to ensure the desired performance of the sensing beampattern in radar-based ISAC systems under uncertain target location, \cite{Fan:18,Xiang:20,Xiang:22b} employ beampattern matching as a sensing metric. On the other hand, the signal-to-interference-plus-noise ratio (SINR) has been predominantly used as a measure of sensing signal quality and target tracking performance in various scenarios, including multiple-target single-user \cite{Chen:22a}, single-target multi-user \cite{Yuanhan:22}, and multiple-target multi-user systems \cite{pritzker2022transmit}, as well as cooperative ISAC networks \cite{Na:22} and two-cell interfering ISAC networks \cite{Yiqing:22}. Due to the significant mathematical disparities among these sensing constraints, existing ways to manage these sensing requirements also vary substantially, ranging from successive convex approximation (SCA) \cite{Fan:18,Xiang:22b,Yiqing:22}, semidefinite relaxation (SDR) \cite{Xiang:20,Yuanhan:22,Na:22} to Taylor expansion \cite{Chen:22a}.

The diversity in the objective functions and sensing constraints results in numerous combinations of different scenarios. Each of them typically requires separate investigation in the current practice. Not only does this demand considerable efforts, but the variety of ideas and techniques used across different studies also makes direct comparison between scenarios challenging. This paper aims to develop a unified optimization framework for these different secure ISAC scenarios. Such a framework not only enhances the adaptability of ISAC solutions across diverse sensing scenarios, but also enables fair and straightforward comparisons among different scenarios with minimal additional effort.

More specifically, to combine the worst case rate maximization and weighted sum rate maximization into a unified formulation, this paper for the first time proposes a general max-min problem which reduces to worst case rate maximization if the constraint of the inner minimization is a simplex. On the other hand, if the constraint of the inner minimization variables is a singleton, the corresponding problem becomes the weighted sum rate maximization. Furthermore, it is found that under such general secrecy rate objective function, the secrecy outage constraint can be tackled without approximation. This results in tightly realized target probability, which in turns leave more resource for further optimizing the secrecy rate. 

To accommodate various sensing performance constraints in a unified way, we introduce a penalty variable and move the sensing constraint to the objective function. In addition to facilitating subsequent algorithm derivation, with the sensing constraint moved to the objective function, the feasible region is enlarged and better iteration trajectories are expected to be found. This is verified theoretically that the transformed problem share the same optimal solution as in the original problem with explicit sensing constraint, and a stationary point of the former must be at least a stationary point of the latter. 

After the above transformation, since the remaining constraints are decoupled in the optimization variables, we further propose a projected gradient based alternative optimization (AO) with its solution quality guarantee theoretically proved. The advantage of the proposed unified framework is that when the sensing metric or objective function changes, only the gradient of the sensing metric or the projection of the auxiliary variables needs to be modified, rather than redesigning the whole algorithm. Leveraging the highly efficient gradient descent under simple constraints, the proposed method maintains a low computational complexity in each iteration. Moreover, since the proposed algorithm is non-monotonic, it can traverse the valleys of the objective function’s rugged surface to discover better solutions. Simulation results show that the proposed framework outperforms other existing algorithms in both running time and optimized secrecy rate under various sensing metric settings.

The contributions of this paper are summarized as follows.

\begin{enumerate}[leftmargin=0.5cm]
	\item We propose a framework incorporating the redundant information technique and artificial noise to enhance security of ISAC under covert eavesdroppers. We only assume the knowledge of distribution of eavesdroppers' channels, but not their realizations. This leads to an intricate trade-off between secrecy outage probability (SOP), secrecy rate optimization and sensing performance.
	
	\item The proposed framework holistically maximizes the worst-case user secrecy rate and the sum secrecy rate across multiple users, with both cases subsumed within a unified objective function. Moreover, the framework accurately controls secrecy outage probabilities while accommodating a broad spectrum of sensing constraints.
	
	\item Unlike traditional approaches that rely on convex relaxation to handle sensing constraints, we propose a gradient-based iterative algorithm. This method leverages gradient information, combined with bisection methods and projection operations, to achieve fast convergence. Additionally, we theoretically establish the solution quality of the proposed algorithm.
	
	\item We demonstrate the proposed unified framework with three different ISAC constraints as examples. Extensive simulation results show that the proposed framework consistently achieves higher optimized secrecy rates under various sensing constraints compared to existing methods. These results underscore its superiority and versatility in secure ISAC systems.
\end{enumerate} 

The rest of the paper is organized as follows. The secure ISAC system model is introduced and the worst user secrecy rate maximization is formulated in Section II. The problem formulation is generalized to unify worst user rate maximization and sum rate maximization under various sensing metrics in Section III. The proposed optimization framework and its solution quality analysis are presented in Section IV. The specific computation steps of the proposed framework under different sensing metrics are derived in Section V. Simulation results are provided in Section VI and conclusions are drawn in Section VII.

Notation: Symbols for vectors (lower case) and matrices (upper case) are presented in boldface. The transpose and conjugate transpose of a matrix are denoted by ${\left(  \cdot  \right)^T}$ and ${\left(  \cdot  \right)^H}$, respectively. The notation $\left\|  \cdot  \right\|_2$ represents the L2 norm of a vector, while ${{\bm I}_N}$ is the identity matrix with dimension $N \times N$. A complex-valued vector ${\bm x}\sim{\rm{{\cal C}{\cal N}}}\left( {{\bm 0},{\bm C}} \right)$ means ${\bm x}$ is complex Gaussian distributed with zero mean and covariance matrix $\bm C$.

\section{System Model and worst user secrecy rate maximization}

\begin{figure}[t]
	\centering
	\includegraphics[width=0.5\linewidth]{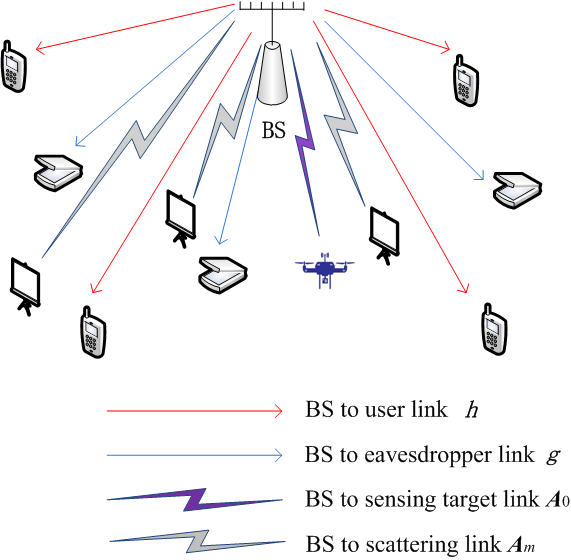}
	\caption{Secure ISAC network model}
	\label{fig1}
\end{figure}
We consider a ISAC network with a BS equipped with $N$ antennas as shown in Fig. \ref{fig1}. The system includes $I$ communication users, $J$ eavesdroppers, each with a single antenna, and the channels from the BS to the $i^{th}$ user and the $j^{th}$ eavesdropper are denoted by ${{\bm h}_i} \in {{\mathbb C}^{N \times 1}}$ and ${{\bm g}_j} \in {{\mathbb C}^{N \times 1}}$, respectively. Besides, this network has a sensing target and $M$ scatterings. In practice, the BS only has knowledge ${{\bm g}_j} \sim {\cal CN}\big( {{\bm 0},\rho _j^2{{\bm I}_N}} \big)$, but not the exact value of ${{\bm g}_j}$. The $\rho _j$ here represents the large-scale fading coefficient of the $j^{th}$ eavesdropper and it could be known if we know its approximate location. This assumption is commonly used in PLS to model the passive eavesdropping channels \cite{Wei:17,Xi:15,Qi:14,Hui-Ming:15}. Furthermore, to avoid electromagnetic coupling, the transmit antennas are typically spaced by at least half a wavelength. Therefore, the signal paths arriving at or leaving from each antenna have negligible correlation, which leads to 0 in the off-diagonal elements of the channel covariance matrix. On the other hand, the assumption of the same diagonal elements of the covariance matrix is often valid when the BS is equipped with a uniform linear array (ULA) or uniform planar array (UPA), where all antennas have the same characteristics and are spaced symmetrically \cite{Giovanni:25}. The above deployment ensures that no single antenna consistently experiences a stronger average signal than the others.

Let ${s_i} \in {\mathbb C}$ be the information symbol intended for the $i^{th}$ user, with ${s_i} \sim {\cal CN}\big( {0,1} \big)$ and is independent among users. With the beamforming vectors $\left\{ {{{\bm w}_i}} \right\}_{i = 1}^I$ for communication, artificial beamformer $\bm w_0$ and corresponding AN ${s_0} \sim {\cal CN}\big( {0,1} \big)$ for sensing and countering eavesdroppers, the transmitted signal of the BS is given by $\sum\nolimits_{l = 1}^I {{{\bm w}_l}{s_l}}  + {{\bm w}_0}s_0$. Correspondingly, the received signals at the $i^{th}$ user and the $j^{th}$ Eve are respectively given by
\begin{equation}\label{eq:1}
		{y_i} = {\bm h}_i^T\big( {\textstyle\sum\nolimits_{l = 1}^I {{{\bm w}_l}{s_l}}  + {{\bm w}_0}s_0} \big) + {n_i},
	\end{equation}
	\begin{equation}\label{eq:2}
		{e_j} = {\bm g}_j^T\big( {\textstyle\sum\nolimits_{l = 1}^I {{{\bm w}_l}{s_l} + {{\bm w}_0}s_0} } \big) + {\tilde n_j},
\end{equation}

\noindent where ${n_i}\sim{\cal CN}\big( {0,\sigma _i^2} \big)$ and ${\tilde n_j}\sim{\cal CN}\big( {0,\varsigma _j^2} \big)$ represent the Gaussian noises at the $i^{th}$ user and the $j^{th}$ Eve with variance $\sigma _i^2$ and $\varsigma _j^2$, respectively. 

Moreover, the echo signal received at the BS from the sensing target and scatterings is given by
\begin{equation}\label{eq:3}
		{\bm r} = \textstyle\sum\nolimits_{m = 0}^M {{{\bm A}_m}\big( {\sum\nolimits_{l = 1}^I {{{\bm w}_l}{s_l} + {{\bm w}_0}s_0} } \big)}  + {{\bm n}_{BS}},
\end{equation}

\noindent where ${{\bm n}_{BS}}\sim{\cal CN}\big( {0,\sigma _{BS}^2{\bm I}_N} \big)$ represents the Gaussian noises at the BS with variance $\sigma _{BS}^2$, $\left\{ {{{\bm A}_m}} \right\}_{m = 0}^M$ are the sensing channels of the target ($\bm A_0$) and scatterings ${\left\{ {{{\bm A}_m}} \right\}_{m \ne 0}}$. In particular, ${{\bm A}_m} = {\beta _m}{\bm a}\left( {{\theta _m}} \right){{\bm a}^H}\left( {{\theta _m}} \right)$, where $\beta_m$ is the round-trip channel coefficient including the path loss and the radar cross-section (RCS) for the target ($m=0$) or scatterings ($m \neq 0$), and
\begin{equation}\label{eq:4}
	{\bm a}\left( \theta_m  \right) = {\big[ {1,{e^{j2\pi \zeta \sin \theta_m }}, \ldots ,{e^{j2\pi \left( {N - 1} \right)\zeta \sin \theta_m }}} \big]^T}
\end{equation}
is the array steering vector of the BS with the angle from the BS to the sensing target denoted as ${\theta _0}$ and that to the $m^{th}$ scattering given by ${{\theta _m}}$. The symbol $\zeta$ is the adjacent antennas spacing at the BS normalized by the wavelength~\cite{Boxiang:24}. Note that the channels ${\bm h}_i$, $\bm A_0$ are estimated in the pilot transmission phase \cite{Zhenyao:23,Ziang:23} and the scattering channels ${\left\{ {{{\bm A}_m}} \right\}_{m \ne 0}}$ in the environment can be estimated based on field measurements or radio maps \cite{Suzhi:19,Yang:25}.

With a receive filter $\bm u \in {\mathbb C}^{N \times 1}$ adopted at the BS, the sensing SINR at the BS is
\begin{equation}\label{eq:5}
	{\gamma _{Rad}} = \frac{{\sum\nolimits_{l = 0}^I {\left| {{{\bm u}^T}{{\bm A}_0}{{\bm w}_l}} \right|_2^2} }}{{\sum\nolimits_{l = 0}^I {\sum\nolimits_{m = 1}^M {{{\left| {{{\bm u}^T}{{\bm A}_m}{{\bm w}_l}} \right|}^2}} }  + \sigma _{BS}^2\left\| {\bm u} \right\|_2^2}}.
\end{equation}
On the other hand, the achievable rate of the $j^{th}$ Eve for monitoring the $i^{th}$ user (also known as the eavesdropping rate) is given by 
\begin{equation}\label{eq:6}
	{C_{ji}} = {\log _2}\left( {1 + \frac{{{{\left| {{\bm g}_j^T{{\bm w}_i}} \right|}^2}}}{{\sum\nolimits_{l = 0,l \ne i}^I {{{\left| {{\bm g}_j^T{{\bm w}_l}} \right|}^2}}  + \varsigma _j^2}}} \right).
\end{equation}
Since the BS does not have perfect knowledge of the eavesdroppers' channels, the value of $C_{ji}$ is uncertain from the perspective of the BS. Consequently, a secrecy outage event occurs at the BSs when $C_{ji}$ exceeds the redundancy rate of the $i^{th}$ user, denoted by $D_{ji}$, and SOP of the user $i$ due to Eve $j$ is expressed as
\begin{equation}\label{eq:7}
	{\rm SOP}\;:\;p_{ji}^{so} = \Pr \left\{ {{D_{ji}} \le {C_{ji}}} \right\}.
\end{equation}

Our objective is to maximize the secrecy rate of the worst communication user while maintaining the secrecy outage probability and the sensing performance of the target. Correspondingly, the optimization problem is formulated as
\begin{subequations}\label{eq:8}
	\begin{align}
		\mathop {\max }\limits_{\left\{ {{{\bm w}_l}} \right\}_{l = 0}^I,\left\{ {{D_{ji}}} \right\},{\bm u}} \mathop {\min }\limits_{j,i}\;\;& \big\{ {{{ {{{\rm{{\cal R}}}_i}\big( \big\{ {{{\bm w}_l}} \big\}_{l = 0}^I \big) - {D_{ji}}} } }} \big\}\label{eq:8a}\\
		s.t.\;\;&p_{ji}^{so} \le {\eta _i},\;\;\forall j,i\label{eq:8b}\\
		&\sum\nolimits_{l = 0}^I {{{\left\| {{{\bm w}_l}} \right\|}_2^2}}  \le {P_{BS}},\label{eq:8c}\\
		&{\gamma _{Rad}} \ge {\Gamma },\label{eq:8d}
	\end{align}
\end{subequations}
where ${{\rm{{\cal R}}}_i}\big( \left\{ {{{\bm w}_l}} \right\}_{l = 0}^I \big) = {\log _2}\left( {1 + \frac{{{{\left| {{\bm h}_i^T{{\bm w}_i}} \right|}^2}}}{{\sum\nolimits_{l = 0,l \ne i}^I {{{\left| {{\bm h}_i^T{{\bm w}_l}} \right|}^2}}  + \sigma _i^2}}} \right)$. Constraint \eqref{eq:8b} states that the secrecy outage probability of the $i^{th}$ user monitored by the $j^{th}$ eavesdropper should be smaller than a target threshold ${\eta _i} \in \big( {0,1} \big)$. Constraint \eqref{eq:8c} is the total power budget of the BS and \eqref{eq:8d} requires the sensing SINR to be at least $\Gamma$ to guarantee sensing performance. Notice that the SOP constraint \eqref{eq:8b} is related to the beamformers inside the probability function. But the following lemma could be adopted to convert it into a deterministic form.
\begin{myLem} \cite[Lem 2]{Yang:20}
	Suppose ${{\bm g}_{j}} \sim {\rm{{\cal C}{\cal N}}} \big( {{\bm 0},\rho _{j}^2{{\bm I}_N}} \big), \forall j$, the SOP constraint \eqref{eq:8b} can be equivalently transformed into
	\begin{equation}\label{eq:9}
		\begin{split}
			&\ln {\eta _i} + \frac{{\left( {{2^{{D_{ji}}}} - 1} \right)}}{{\left( {{{\rho _j^2} \mathord{\left/
								{\vphantom {{\rho _j^2} {\varsigma _j^2}}} \right.
								\kern-\nulldelimiterspace} {\varsigma _j^2}}} \right)\left\| {{{\bm w}_i}} \right\|_2^2}} + \sum\limits_{l = 0,l \ne i}^I {\ln \left( {1 + \left( {{2^{{D_{ji}}}} - 1} \right)\frac{{\left\| {{{\bm w}_l}} \right\|_2^2}}{{\left\| {{{\bm w}_i}} \right\|_2^2}}} \right)} \\
			&\ge 0,
		\end{split}
	\end{equation}
	where ${{\rho _j^2} \mathord{\left/
			{\vphantom {{\rho _j^2} {\varsigma _j^2}}} \right.
			\kern-\nulldelimiterspace} {\varsigma _j^2}}$ is interpreted as the statistical signal-to-noise ratio (SNR) for the $j^{th}$ eavesdropper.
\end{myLem}
Although the probability constraint \eqref{eq:8b} can be equivalently converted into a deterministic form in \eqref{eq:9}, it is still complex since the the beamformers are coupled in a logarithm and fraction form. To handle this transformed SOP constraint, \cite{Wei-Chiang:15,Sheng:18,Zhichao:18} introduce $J \times I$ auxiliary constraints for substituting ${{\left( {{2^{{D_{ji}}}} - 1} \right)} \mathord{\left/
		{\vphantom {{\left( {{2^{{D_{ji}}}} - 1} \right)} {\left\| {{{\bm w}_i}} \right\|_2^2}}} \right.
		\kern-\nulldelimiterspace} {\left\| {{{\bm w}_i}} \right\|_2^2}}$ with a new auxiliary variable $t_{ji}$ and then use SCA to convexify all these introduced constraints. Another work \cite{Wai-Yiu:23} uses the matrix form of the beamformers ${{\bm W}_i} = {{\bm w}_i}{\bm w}_i^H$, and converts all the quadratic term $\left\| {{{\bm w}_l}} \right\|_2^2$ as a linear combination of ${\rm Tr}\left( {{{\bm W}_l}} \right)$. Then, SDR can be adopted through further introducing convexified auxiliary constraints and temporarily discarding the rank-1 constraint. In general, these existing approaches involve exchanging \eqref{eq:9} for a series of simpler, though more conservative, convexified constraints. Unfortunately, this often results in allocating limited resources to merely surpass the target outage probability, rather than optimizing the objective function~\cite{Kun-Yu:11}. 

To this end, we introduce a method for handling \eqref{eq:9} without approximation or relaxation, thus tightly realize the target outage probability. In particular, it is recognized that the left hand side of \eqref{eq:9} is a monotonic increasing function of $D_{ji}$. On the other hand, to maximize the secrecy rate \eqref{eq:8a}, $D_{ji}$ should be minimized. Therefore, the optimal value of $D_{ji}$ that maximizes the secrecy rate while maintaining the SOP constraint is given by \eqref{eq:9} with equality holds. As a result, the optimal $D_{ji}$ can be found by the bisection method when $\left\{ {{{\bm w}_l}} \right\}_{l = 0}^I$ are fixed. By denoting ${D_{ji}}$ as ${D_{ji}}\big( \left\{ {{{\bm w}_l}} \right\}_{l = 0}^I \big)$, problem \eqref{eq:8} becomes
\begin{equation}\label{eq:10}
	\begin{split}
		\mathop {\max }\limits_{\left\{ {{{\bm w}_l}} \right\}_{l = 0}^I,{\bm u}} \mathop {\min }\limits_i \;\;&\big\{ {{{\rm{{\cal R}}}_i}\big( {\big\{ {{{\bm w}_l}} \big\}_{l = 0}^I} \big) - \mathop {\max }\limits_j {D_{ji}}\big( {\big\{ {{{\bm w}_l}} \big\}_{l = 0}^I} \big)} \big\}\\
		s.t.\;\;&\sum\nolimits_{l = 0}^I {{{\left\| {{{{\bm w}}_l}} \right\|}_2^2}}  \le {P_{BS}},\\
		&{\gamma _{Rad}} \ge {\Gamma }.
	\end{split}
\end{equation}
From \eqref{eq:10}, it is clear that the beamformers serve both purposes of enhancing security and target sensing. It can be anticipated that a higher sensing SINR requirement $\Gamma$ would decrease the secrecy rate, and vice versa.

Through expressing ${D_{ji}}$ as an implicit function of beamformers $\left\{ {{{\bm w}_l}} \right\}_{l = 0}^I$, the SOP constraint does not explicitly appear in \eqref{eq:10}. Since the optimal ${{D_{ji}}\big( \left\{ {{{\bm w}_l}} \right\}_{l = 0}^I \big)}$ is given by \eqref{eq:9} with equality holds, a large ${{{\rho _j^2} \mathord{\left/
			{\vphantom {{\rho _j^2} {\varsigma _j^2}}} \right.
			\kern-\nulldelimiterspace} {\varsigma _j^2}}}$ would result in a large redundancy rate ${{D_{ji}}\big( \left\{ {{{\bm w}_l}} \right\}_{l = 0}^I \big)}$. Therefore, the largest redundancy rate ${D_i}\big( {\left\{ {{{\bm w}_l}} \right\}_{l = 0}^I} \big): = \mathop {\max }\limits_j \big\{ {{D_{ji}}\big( {\left\{ {{{\bm w}_l}} \right\}_{l = 0}^I} \big)} \big\}$ of user $i$ due to all eavasdroppers satisfies
\begin{equation}\label{eq:11}
	\begin{split}
		&\ln {\eta _i} + {{\left( {{2^{{D_i}\left( {\left\{ {{{\bm w}_l}} \right\}_{l = 0}^I} \right)}} - 1} \right)} \mathord{\left/
				{\vphantom {{\left( {{2^{{D_i}\left( {\left\{ {{{\bm w}_l}} \right\}_{l = 0}^I} \right)}} - 1} \right)} {\left[ {\mathop {\max }\limits_j \left( {{{\rho _j^2} \mathord{\left/
											{\vphantom {{\rho _j^2} {\varsigma _j^2}}} \right.
											\kern-\nulldelimiterspace} {\varsigma _j^2}}} \right)\left\| {{{\bm w}_i}} \right\|_2^2} \right]}}} \right.
				\kern-\nulldelimiterspace} {\left[ {\mathop {\max }\limits_j \left( {{{\rho _j^2} \mathord{\left/
								{\vphantom {{\rho _j^2} {\varsigma _j^2}}} \right.
								\kern-\nulldelimiterspace} {\varsigma _j^2}}} \right)\left\| {{{\bm w}_i}} \right\|_2^2} \right]}}\\
		&+ \sum\limits_{l = 0,l \ne i}^I {\ln \left( {1 + \left( {{2^{{D_i}\left( {\left\{ {{{\bm w}_l}} \right\}_{l = 0}^I} \right)}} - 1} \right)\frac{{\left\| {{{\bm w}_l}} \right\|_2^2}}{{\left\| {{{\bm w}_i}} \right\|_2^2}}} \right)}  = 0,
	\end{split}
\end{equation}
which coincides with the intuition that the eavesdropper with the highest statistical SNR is the one that determines what redundancy rate a user needs in order to maintain the outage probability. Then problem \eqref{eq:10} can be simplified as
\begin{subequations}\label{eq:12}
	\begin{align}
		\mathop {\max }\limits_{{\left\{ {{{\bm w}_l}} \right\}_{l = 0}^I},{\bm u}} \mathop {\min }\limits_{i}\;\;& \big\{ {{{ {{{\rm{{\cal R}}}_i}\big( {\big\{ {{{\bm w}_l}} \big\}_{l = 0}^I} \big) - {D_{i}\big( {\big\{ {{{\bm w}_l}} \big\}_{l = 0}^I} \big)}} } }} \big\}\label{eq:12a}\\
		s.t.\;\;&\sum\nolimits_{l = 0}^I {{{\left\| {{{\bm w}_l}} \right\|}_2^2}}  \le {P_{BS}},\label{eq:12b}\\
		&{\gamma _{Rad}} \ge {\Gamma }.\label{eq:12c}
	\end{align}
\end{subequations}

According to the formulation of sensing SINR \eqref{eq:5}, the amplitude of the receive filter $\bm u$ does not affect the value of ${\gamma _{rad}}$, therefore, the sensing constraint \eqref{eq:12c} can be equivalently transformed into
\begin{subequations}\label{eq:13}
	\begin{align}
		&\sum\nolimits_{l = 0}^I {\left| {{{\bm u}^T}{{\bm A}_0}{{\bm w}_l}} \right|_2^2}  - \Gamma \sum\nolimits_{l = 0}^I {\sum\nolimits_{m = 1}^M {{{\left| {{{\bm u}^T}{{\bm A}_m}{{\bm w}_l}} \right|}^2}} }  - \Gamma \sigma _{BS}^2 \nonumber \\
		&\ge 0,\label{eq:13a}\\
		&\left\| {\bm u} \right\|_2^2 = 1.\label{eq:13b}
	\end{align}
\end{subequations}
Hence, the worst case secrecy rate maximization under sensing SINR constraint is equivalent to
\begin{subequations}\label{eq:14}
	\begin{align}
		\mathop {\max }\limits_{{\left\{ {{{\bm w}_l}} \right\}_{l = 0}^I},{\bm u}} \mathop {\min }\limits_i \;\;&\left\{ {{{\rm{{\cal R}}}_i}\left( {\left\{ {{{\bm w}_l}} \right\}_{l = 0}^I} \right) - {D_i}\left( {\left\{ {{{\bm w}_l}} \right\}_{l = 0}^I} \right)} \right\}\label{eq:14a}\\
		s.t.\;\;&\sum\nolimits_{l = 0}^I {\left\| {{{\bm w}_l}} \right\|_2^2}  \le {P_{BS}},\label{eq:14b}\\
		&\eqref{eq:13a},\eqref{eq:13b}. \nonumber
	\end{align}
\end{subequations}
It is emphasized that \eqref{eq:14} covers the common worst user's achievable rate maximization if eavesdroppers do not exist, which corresponds to ${D_i}\left( {\left\{ {{{\bm w}_l}} \right\}_{l = 0}^I} \right) \equiv 0$. Moreover, in ISAC, there exists sensing performance criteria other than SINR constraint. The above two observations inspire us to extend \eqref{eq:14} in the following section.

\textbf{Remark 1 (SOP vs intercept probability):} Beside SOP, another commonly used outage probability is the intercept probability \cite{Jiaxin:15,Yulong:13}, which is given by
\begin{equation}
	\Pr \big\{ {{{{{\cal R}}}_i}\big( {\left\{ {{{\bm w}_l}} \right\}_{l = 0}^I} \big) \leqslant {C_{ji}}} \big\}.
\end{equation}
It is a metric to quantify the likelihood of a successful interception of the entire information transmitted with ${{\rm{{\cal R}}}_i}\big( \left\{ {{{\bm w}_l}} \right\}_{l = 0}^I \big)$. On the other hand, SOP in \eqref{eq:7} measures the probability of the information leakage of the $i^{th}$ user. When the eavesdropping rate $C_{ji}$ exceeds the redundant information $D_{ji}$, it means the confidentiality of the transmitted information cannot be guaranteed. However, it is important to note that information leakage does not necessarily mean that all useful information has been intercepted by the eavesdroppers. In general, since $D_{ji}$ is the amount of redundant information added to a message to secure it from eavesdroppers, ${{{{{\cal R}}}_i}\left( {\left\{ {{{\bm w}_l}} \right\}_{l = 0}^I} \right) > {D_{ji}}}$, which makes
\begin{equation}\label{eq:Remark1}
	\Pr \left\{ {{D_{ji}} \leqslant {C_{ji}}} \right\} \geqslant \Pr \left\{ {{{{{\cal R}}}_i}\left( {\left\{ {{{\bm w}_l}} \right\}_{l = 0}^I} \right) \leqslant {C_{ji}}} \right\}.
\end{equation}
The implication of \eqref{eq:Remark1} is that imposing constraint on SOP would also restrict intercept probability. In another perspective, SOP provides a more stringent measure of security. Together with the fact that employing redundant information is a common strategy in secure and covert transmission \cite{Baogang:24,Jingke:25,Xuanxuan:19,Shengbin:24,Qunshu:25}, our primary focus is on the SOP.

\section{Extension to general sensing constraint in ISAC}
As the sensing target and the scatterings are not directly related to the communication users and eavesdroppers, with the introduction of auxiliary variables $\left\{ {{y_i}} \right\}_{i = 1}^I$, we can extend \eqref{eq:14} to a generalized rate optimization in ISAC:

\begin{subequations}\label{eq:19}
	\begin{align}
		\mathop {\max }\limits_{\left\{ {{{\bm w}_l}} \right\}_{l = 0}^I,{\bm v}}\mathop {\min }\limits_{\left\{ {{y_i}} \right\}_{i = 1}^I}& \sum\nolimits_{i = 1}^I {{y_i}{\psi _i}\left( {\left\{ {{{\bm w}_l}} \right\}_{l = 0}^I} \right)} \label{eq:19a}\\
		s.t.\;\;&\sum\nolimits_{l = 0}^I {\left\| {{{\bm w}_l}} \right\|_2^2}  \le {P_{BS}},\label{eq:19b}\\
		&S\left( {{\left\{ {{{\bm w}_l}} \right\}_{l = 0}^I},{\bm v}} \right) \ge 0,\label{eq:19c}\\
		&\bm v \in \Omega_{\bm v},\label{eq:19d}\\
		& \left\{ {{y_i}} \right\} \in {\Omega _{\bm y}},\label{eq:19e}
	\end{align}
\end{subequations}
where ${\psi _i}\big( {\big\{ {{{\bm w}_l}} \big\}_{l = 0}^I} \big)$ could be the secrecy rate ${{\rm{{\cal R}}}_i}\big( {\left\{ {{{\bm w}_l}} \right\}_{l = 0}^I} \big) - {D_i}\big( {\left\{ {{{\bm w}_l}} \right\}_{l = 0}^I} \big)$ if eavesdroppers are present or the achievable rate ${{\rm{{\cal R}}}_i}\big( {\left\{ {{{\bm w}_l}} \right\}_{l = 0}^I} \big)$ if there is no eavesdropper, $S\big( {{\left\{ {{{\bm w}_l}} \right\}_{l = 0}^I},{\bm v}} \big)$ is the sensing performance metric of interest, ${\bm v}$ contains other optimization variables related to the sensing performance and $\Omega_{\bm v}$ is the feasible region of $\bm v$. Furthermore, ${{\Omega _{\bm y}}}$ is the feasible set of the auxiliary variables $\left\{ {{y_i}} \right\}_{i = 1}^I$. Without loss of generality, we assume that $\big\{ {{\psi _i}\big( {\big\{ {{{\bm w}_l}} \big\}_{l = 0}^I} \big)} \big\}_{i = 1}^I$ and $S\big( {{\left\{ {{{\bm w}_l}} \right\}_{l = 0}^I},{\bm v}} \big)$ are twice differentiable.

Problem \eqref{eq:19} covers two widely used objective functions. 

\noindent \underline{\emph{Objective 1: Worst user secrecy rate.}} In the case, the constraint on $\bm y$ is a simplex ${\Omega _{\bm y}}: = \big\{ {\sum\nolimits_{i = 1}^I {{y_i}}  = 1,{y_i} \ge 0,\forall i} \big\}$. In order to minimize the summation of the rates under this ${\Omega _{\bm y}}$, the inner minimization with respect to (w.r.t.) $\left\{ {{y_i}} \right\}$ would allocate zero to the non-minimal rates and one to the minimal rate so that the summation of rates would reach its minimal solution, i.e., $\mathop {\min }\limits_i {\psi _i}\big( {\left\{ {{{\bm w}_l}} \right\}_{l = 0}^I} \big)$ as in worst case rate maximization \eqref{eq:14}. An advantage of the formulation \eqref{eq:19} under simplex constraint of $\left\{ {{y_i}} \right\}_{i = 1}^I$ over the worst user rate maximization \eqref{eq:14} is that the inner minimization is changed from a discrete user selection problem into a minimization problem with continuous variables $\left\{ {{y_i}} \right\}$, which allows a continuous iteration of these variables (more details will be presented in the next section). 

\noindent \underline{\emph{Objective 2: Weighted sum secrecy rate.}}
In this case, ${\Omega _{\bm y}}: = \left\{ {{y_i} = {{\tilde y}_i},\forall i} \right\}$ is a singleton set. Then the inner minimization w.r.t. $\left\{ {{y_i}} \right\}$ in \eqref{eq:19} can be skipped. As a special case, when ${{{\tilde y}_i} = 1}$ for all $i$, it becomes the sum rate maximization problem. Also, if ${{{\tilde y}_i}}$ is one for one user and zero for other users, problem \eqref{eq:19} would turn into rate maximization of a single user. This makes \eqref{eq:19a} a generalization to include both worst user rate maximization and weighted sum rate maximization for the communication users. Notice that since SOP constraint does not explicitly appear in \eqref{eq:19}, for an ISAC system without eavesdroppers, ${\psi _i}\big( {\left\{ {{w_l}} \right\}_{l = 0}^I} \big)$ can be any utility beyond the secrecy rate and achievable rate. However, since we focus on secure ISAC, we do not pursue this direction further in this paper.

Formulation \eqref{eq:19} not only extends the objective function but also generalizes the sensing metric. While our goal is not to exhaustively list all possible sensing criteria, below are a few popular sensing metrics used in existing research.

\noindent \underline{\emph{Metric 1: Sensing SINR.}} This is the case we discussed in Section II. Correspondingly, $S\big( {{\left\{ {{{\bm w}_l}} \right\}_{l = 0}^I},{\bm v}} \big)$ is given by the left hand side of \eqref{eq:13a} with $\bm v=\bm u$ and ${{\Omega _{\bm v}}}$ is the feasible region given by the constraint \eqref{eq:13b}.

\noindent \underline{\emph{Metric 2: Beampattern matching.}} The mean squared error (MSE) between the sensing beampattern and the ideal beampattern can be employed as a sensing criterion, which is written as \cite{Siyi:25,Peng:22}
\begin{equation}\label{eq:20}
	\begin{split}
		&MSE\left( {\left\{ {{{\bm w}_l}} \right\}_{l = 0}^I} \right) \\
		&= \frac{1}{T}\sum\limits_{t = 1}^T {{{\left| {{P_d}\left( {{\theta _t}} \right) - {{\bm a}^H}\left( {{\theta _t}} \right)\left( {\sum\limits_{l = 0}^I {{{\bm w}_l}{\bm w}_l^H} } \right){\bm a}\left( {{\theta _t}} \right)} \right|}^2}},
	\end{split}
\end{equation}
where $T$ represents the number of sampling angles, ${{P_d}\left( {{\theta _t}} \right)}$ is the desired beampattern at the target angle $\theta _t$ and ${\bm a}\left( {{\theta _t}} \right)$ is given in \eqref{eq:4}. Correspondingly, the sensing performance constraint is written as
\begin{equation}\label{eq:21}
	MSE\left( {\left\{ {{{\bm w}_l}} \right\}_{l = 0}^I} \right) \le \gamma,
\end{equation}
where $\gamma$ is the allowable MSE. In this case, $S\big( {\left\{ {{{\bm w}_l}} \right\}_{l = 0}^I} \big) = \gamma  - MSE\big( {\left\{ {{{\bm w}_l}} \right\}_{l = 0}^I} \big)$ and ${\Omega _{\bm v}} = \emptyset $.

\noindent \underline{\emph{Metric 3: Detection probability.}} It is known that with the receive filter ${\bm u} \in {{\mathbb C}^{N \times 1}}$ orthogonal to the steering vectors ${{{\left\{ {{\bm a}\left( {{\theta _m}} \right)} \right\}}_{m \ne 0}}}$ of all scatterings, the interferences from the scatterings can be suppressed. 
With the derivations in \cite{Jiancheng:23}, the detection probability ${P_{D}}$ of the sensing target for a given false-alarm probability ${P_{FA}}$ is given by
\begin{equation}\label{eq:23}
	\begin{split}
	&{P_D} = \\
	&Q\left( {{Q^{ - 1}}\left( {{P_{FA}}} \right) - \sqrt {\frac{{2\beta _0^2{{\bm a}^H}\left( {{\theta _0}} \right)\left( {\sum\limits_{l = 0}^I {{{\bm w}_l}{\bm w}_l^H} } \right){\bm a}\left( {{\theta _0}} \right)}}{{\sigma _{BS}^2}}} } \right),
	\end{split}
\end{equation}
where $Q\left(  \cdot  \right)$ is the right-tailed distribution function of the standard Gaussian distribution and ${Q^{ - 1}}\left(  \cdot  \right)$ is its inverse function. To maintain a target detection probability $\phi$, we need
\begin{equation}\label{eq:25}
	{P_D} \ge \phi.
\end{equation}
Based on \eqref{eq:23}, constraint \eqref{eq:25} can be reformulated as
\begin{equation}\label{eq:24}
	\sum\nolimits_{l = 0}^I {{\bm w}_l^H{{\bm A}_0}{{\bm w}_l}}  - \left[ {{{\sigma _{BS}^2\left( {{Q^{ - 1}}\left( {{P_{FA}}} \right) - {Q^{ - 1}}\left( {{\phi}} \right)} \right)} \mathord{\left/
				{\vphantom {{\sigma _{BS}^2\left( {{Q^{ - 1}}\left( {{P_{FA}}} \right) - {Q^{ - 1}}\left( {{P_D}} \right)} \right)} {2{\beta _0}}}} \right.
				\kern-\nulldelimiterspace} {2{\beta _0}}}} \right] \ge 0,
\end{equation}
where ${\bm A}_0$ is defined in \eqref{eq:3}. In this case, $S\big( {\left\{ {{{\bm w}_l}} \right\}_{l = 0}^I} \big)$ is given by the left hand side of \eqref{eq:24} and ${\Omega _{\bm v}} = \emptyset $.

\textbf{Remark 2 (Generality of \eqref{eq:19} without eavesdroppers):} The problem formulation \eqref{eq:19} and thus the subsequent proposed algorithm also cover the cases of worst user sum rate and weighted user sum rate maximization without eavesdroppers. This can be easily achieved by setting ${\eta _i} = 1$ in the SOP constraint, which will automatically lead to ${D_i}\big( {\left\{ {{{\bm w}_l}} \right\}_{l = 0}^I} \big) = 0$ in \eqref{eq:11}. While this case is not the focus of this paper as no PLS is provided, the developed framework also unifies the treatment of ordinary ISAC under various sensing constraints.

\section{Proposed framework for solving the general max-min problem \eqref{eq:19}}
Since \eqref{eq:19} includes worst user rate or sum rate in the objective and various sensing metrics in the constraint as special cases, it is desirable to have a unified framework for solving it. For tackling the general max-min formulation \eqref{eq:19}, existing methods would firstly transform the max-min into a pure maxmization problem, either through introducing a single auxiliary variable to replace the inner minimization in the objective function and then put the inner minimization as constraints \cite{Dongfang:22,Boxiang:24,Rui:24}, or using the penalty reformulation \cite{Wei-Chiang:15}. In general, these methods either face difficulty in setting an appropriate objective function for the outer layer variables~\cite{Liu:24}, or losing the solution quality guarantee due to manually tuning of penalty parameter.

Furthermore, the generally non-convex or even highly complex $S\big( {\left\{ {{{\bm w}_l}} \right\}_{l = 0}^I,{\bm v}} \big)$ in the constraint is another challenge in handling \eqref{eq:19}. Existing dominant methods for handling a non-convex constraint $S\big( {\left\{ {{{\bm w}_l}} \right\}_{l = 0}^I,{\bm v}} \big)$ would involve convexifying it. However, since the sensing constraint in \eqref{eq:19} is in a general form, common tactics such as Taylor expansion or semidefinie relaxization for quadratic constraint may not be applicable. For example, for the beampattern matching constraint, it is not quadratic in $\big\{ {{{\bm w}_l}} \big\}_{l = 0}^I$. In this case, we need to introduce additional conservative auxiliary constraints, which tightens the feasible region and loses convergence points in the excluded region. To develop a solution for the general problem, we will handle the inner minimization and the constraint $S\left( {\left\{ {{{\bm w}_l}} \right\}_{l = 0}^I,{\bm v}} \right)$ at the same time below.

More specifically, rewritting the auxiliary variables $\left\{ {{y_i}} \right\}_{i = 1}^I$ as a vector ${\bm y} := {\left[ {{y_1},{y_2}, \ldots, {y_I}} \right]^T}$ and introducing one more non-negative penalty variable $c$, \eqref{eq:19} can be converted to
\begin{subequations}\label{eq:26}
	\begin{align}
		\mathop {\max }\limits_{\left\{ {{{\bm w}_l}} \right\}_{l = 0}^I,{\bm v}} \mathop {\min }\limits_{c,{\bm y}}\;\;& {{\bm y}^T}{\rm{{\bm \psi}}}\left( {\left\{ {{{\bm w}_l}} \right\}_{l = 0}^I} \right) + cS\left( {\left\{ {{{\bm w}_l}} \right\}_{l = 0}^I,{\bm v}} \right)\label{eq:26a}\\
		s.t.\;\;&\sum\nolimits_{l = 0}^I {\left\| {{{\bm w}_l}} \right\|_2^2}  \le {P_{BS}},\label{eq:26b}\\
		&{\bm v} \in {\Omega _{\bm v}},\label{eq:26c}\\
		&{\bm y} \in {\Omega _{\bm y}},\label{eq:26d}\\
		&c \ge 0,\label{eq:26e}
	\end{align}
\end{subequations}
where ${\bm \psi} \left( {\left\{ {{{\bm w}_l}} \right\}_{l = 0}^I} \right) = {\left[ {{\psi _1}\left( {\left\{ {{{\bm w}_l}} \right\}_{l = 0}^I} \right), \ldots ,{\psi _I}\left( {\left\{ {{{\bm w}_l}} \right\}_{l = 0}^I} \right)} \right]^T}$ is a column vector containing the utilities of all users. The solution equivalence between \eqref{eq:19} and \eqref{eq:26} is given in the following Proposition.
\begin{myPro}
	If $\big\{ {\left\{ {{{\hat {\bm w}}_l}} \right\}_{l = 0}^I,\hat {\bm v},\hat {\bm y},\hat c} \big\}$ is the optimal solution of \eqref{eq:26}, then $\big\{ {\left\{ {{{\hat {\bm w}}_l}} \right\}_{l = 0}^I,\hat {\bm v},\hat {\bm y}} \big\}$ is the optimal solution of \eqref{eq:19}. Conversely, if $\big\{ {\left\{ {{{\hat {\bm w}}_l}} \right\}_{l = 0}^I,\hat {\bm v},\hat {\bm y}} \big\}$ is the optimal solution of \eqref{eq:19}, then $\big\{ {\left\{ {{{\hat {\bm w}}_l}} \right\}_{l = 0}^I,\hat {\bm v},\hat {\bm y},\hat c=0} \big\}$ is the optimal solution of \eqref{eq:26}. Furthermore, if $\big\{ {\left\{ {{\bm w}_l^ * } \right\}_{l = 0}^I,{{\bm v}^ * },{{\bm y}^ * },{c^ * }} \big\}$ is a stationary point of \eqref{eq:26}, $\big\{ {\left\{ {{\bm w}_l^*} \right\}_{l = 0}^I,{{\bm v}^*},\hat {\bm y}^*} \big\}$ is at least a stationary point of \eqref{eq:19}.
\end{myPro}
\emph{Proof:} Due to page limitation, we only provide a sketch of the proof in the following. More details can be found in Appendix A in supplementary material. First, it can be proved that the optimal solution of \eqref{eq:26} must have $\hat cS\big( {\left\{ {{{\hat {\bm w}}_l}} \right\}_{l = 0}^I,\hat {\bm v}} \big) = 0$ and $S\big( {\left\{ {{{\hat {\bm w}}_l}} \right\}_{l = 0}^I,\hat {\bm v}} \big) \ge 0$. This causes the term $cS\big( {\left\{ {{{\bm w}_l}} \right\}_{l = 0}^I,{\bm v}} \big)$ vanishes at the optimal solution of \eqref{eq:26} and makes it equivalent to \eqref{eq:19}. Second, if $\big\{ {\left\{ {{\bm w}_l^*} \right\}_{l = 0}^I,{{\bm v}^*},\hat {\bm y}^*,c^*} \big\}$ is the stationary point of \eqref{eq:26} but $\big\{ {\left\{ {{\bm w}_l^*} \right\}_{l = 0}^I,{{\bm v}^*},\hat {\bm y}^*} \big\}$ is not a stationary point of \eqref{eq:19}, we can construct an sequence gradually approaching $\big\{ {\left\{ {{\bm w}_l^*} \right\}_{l = 0}^I,{{\bm v}^*},\hat {\bm y}^*} \big\}$ and always violates the sensing metric constraint when the points in this sequence is sufficiently close to $\big\{ {\left\{ {{\bm w}_l^*} \right\}_{l = 0}^I,{{\bm v}^*},\hat {\bm y}^*} \big\}$. In this way, we can prove that $\big\{ {\left\{ {{\bm w}_l^*} \right\}_{l = 0}^I,{{\bm v}^*},\hat {\bm y}^*,c^*} \big\}$ is not a stationary point of \eqref{eq:26}. By contradiction, the second part is proved. \hfill $\Box$
\vspace{0.2cm}

Through moving the sensing constraint into the objective function, the functions ${\bm \psi} \big( {\left\{ {{{\bm w}_l}} \right\}_{l = 0}^I} \big)$ and $S\big( {\left\{ {{{\bm w}_l}} \right\}_{l = 0}^I,{\bm v}} \big)$ only exist in the objective function and the constraints become much simpler. Moreover, as $\bm y$ and $c$ are continuous varaibles, an intuitive approach for solving \eqref{eq:26} is to use AO between the inner minimization and the outer maximization. Specifically, at the $k^{th}$ iteration, given the current iteration point $\big\{ {\left\{ {{\bm w}_l^k} \right\}_{l = 0}^I,{{\bm v}^k},{c^k},\bm y^k} \big\}$, we alternatively solve the following two subproblems
\begin{equation}\label{eq:33}
	\begin{split}
		&\left\{ {{c^{k + 1}},{{\bm y}^{k + 1}} } \right\} = \mathop {\arg \min }\limits_{c, {{\bm y}}}\; {{\bm y}^T}{\rm{{\bm \psi}}}\left( {\left\{ {{{\bm w}_l^k}} \right\}_{l = 0}^I} \right) \\ &\;\;\;\;\;\;\;\;\;\;\;\;\;\;\;\;\;\;\;\;\;\;\;\;\;\;\;\;\;\;\;\;\;\;\;\;\;\;+cS\left( {\left\{ {{\bm w}_l^k} \right\}_{l = 0}^I,{{\bm v}^k}} \right)\\
		&\;\;\;\;\;\;\;\;\;\;\;\;\;\;\;\;\;\;\;\;\;\;\;\;\;\;\;\;\;\;\;s.t.\;\;{\bm y} \in {\Omega _{\bm y}}, c \ge 0.
	\end{split}
\end{equation}
\begin{equation}\label{eq:43}
	\begin{split}
		\left\{ {\left\{ {{\bm w}_l^{k + 1}} \right\}_{l = 0}^I,{{\bm v}^{k + 1}}} \right\} = \mathop {\arg \max }\limits_{\left\{ {{{\bm w}_l}} \right\}_{l = 0}^I,{\bm v}}\;\;& {\left( {{{\bm y}^{k + 1}}} \right)^T}{\rm{{\bm \psi}}}\left( {\left\{ {{{\bm w}_l}} \right\}_{l = 0}^I} \right)\\
		&+ {c^{k + 1}}S\left( {\left\{ {{{\bm w}_l}} \right\}_{l = 0}^I,{\bm v}} \right)\\
		s.t.\;\;&\sum\nolimits_{l = 0}^I {\left\| {{{\bm w}_l}} \right\|_2^2}  \le {P_{BS}},\\
		&{\bm v} \in {\Omega _{\bm v}}.
	\end{split}
\end{equation}

Unfortunately, alternatively solving the above two subproblems is not a viable solution. In particular, notice that the objective function in \eqref{eq:33} is linear w.r.t. $c$. Given that the constraint $c \ge 0$ is an unbounded set and $S\big( {\left\{ {{{\bm w}_l}} \right\}_{l = 0}^I,{\bm v}} \big)$ can be negative before convergence, the solution of $c$ in \eqref{eq:33} could be unbounded, which could block subsequent iterations. Furthermore, since the minimization problem \eqref{eq:33} is a strictly linear function of $\bm y$, the optimal solution w.r.t. $\bm y$ may not be unique for the general feasible set ${\Omega _{\bm y}}$ (e.g., ${\Omega _{\bm y}}$ under worst case rate optimization problem with two or more users having the same minimal rate in $\bm \psi \big( {\big\{ {{{\bm w}_l^k}} \big\}_{l = 0}^I} \big)$). Randomly choosing one from the solution set would make the solution sequence has a sudden jump, which imposes challenges to the solution quality analysis. To avoid an unbounded iterate $c^k$ and handle the non-uniqueness of $\bm y^k$, the subproblem for inner minimization should be modified.

\subsection{Regularizing the inner minimization problem}
To deal with the possible unbounded iterate of $c$ and non-unique solution of $\bm y$ at the $k^{th}$ iteration, we propose to modify \eqref{eq:33} as the following regularized minimization problem
\begin{equation}\label{eq:16}
	\begin{split}
		\left\{ {{c^{k + 1}},{{\bm y}^{k + 1}}} \right\}& = \mathop {\arg \min }\limits_{c,{\bm y}}\;\; {{\bm y}^T}{\rm{{\bm \psi}}}\left( {\left\{ {{\bm w}_l^k} \right\}_{l = 0}^I} \right) + \frac{{\lambda _{\bm y}^k}}{2}\left\| {\bm y} \right\|_2^2\\
		&+ cS\left( {\left\{ {{\bm w}_l^k} \right\}_{l = 0}^I,{{\bm v}^k}} \right) + \frac{{\lambda _c^k}}{2}{c^2} + \frac{{{\beta ^k}}}{2}{{\left( {c - {c^k}} \right)}^2}\\
		&\;\;\;\;\;\;\;\;\;\;s.t.\;\;{\bm y} \in {\Omega _{\bm y}},\\
		&\;\;\;\;\;\;\;\;\;\;\;\;\;\;\;\;\;c \in \left[ {0,{C^k}} \right],
	\end{split}
\end{equation}
where ${{\lambda_c ^k}}$, $\lambda_{\bm y} ^k$, $\beta ^k$, $C^k$ are positive regularization parameters. The regularization term ${\beta ^k}{\left( {c - {c^k}} \right)^2}$$/2$ ensures that the next iteration point $c^{k + 1}$ would not go far from the current point $c^k$. The term ${{\lambda _c^k{c^2}} \mathord{\left/
		{\vphantom {{\lambda _c^k{c^2}} 2}} \right.
		\kern-\nulldelimiterspace} 2}$ and $C^k$ guarantees that the solution $c^{k+1}$ of this inner minimization problem would not go to infinity and block the iteration. On the other hand, to make the solution of $\bm y$ unique, we introduce a strongly convex term $\lambda _{\bm y}^k{\left\| {\bm y} \right\|^2}$$/2$. To ensure the subproblem \eqref{eq:16} be consistent with the original subproblem \eqref{eq:33} as the number of iteration increases, we impose $\mathop {\lim }\limits_{k \to \infty } {\lambda_{\bm y} ^k} = \mathop {\lim }\limits_{k \to \infty } {\lambda_c ^k} = \mathop {\lim }\limits_{k \to \infty } {\beta ^k} = 0$ and $\mathop {\lim }\limits_{k \to \infty } {C^k} = \infty $. Although the non-uniqueness of $\bm y$ might occur again when $k$ goes to infinity, this does not matter in practice since the number of iteration would be stopped at a finite number. Obviously, the schedule of the regularization parameters ${{\lambda_{\bm y} ^k}}$, ${{\lambda_c ^k}}$, ${\beta ^k}$ and ${C^k}$ would affect the solution quality of the whole AO iteration. This issue will be analyzed in subsection IV-D. 

\subsection{Finding an ascent point for the outer maximization problem}
Notice that the outer maximization w.r.t. $\big\{ {\left\{ {{{\bm w}_l}} \right\}_{l = 0}^I,{\bm v}} \big\}$ involves the utility functions ${\bm \psi} \big( {\left\{ {{{\bm w}_l}} \right\}_{l = 0}^I} \big)$ and a general sensing constraint $S\big( {\left\{ {{{\bm w}_l}} \right\}_{l = 0}^I,{\bm v}} \big)$, which precludes the possibility of finding a general concave approximation \cite{Wei-Chiang:15,Lei:25,Songtao:20} to the objective function in \eqref{eq:43}. Fortunately, since ${\bm \psi} \big( {\left\{ {{{\bm w}_l}} \right\}_{l = 0}^I} \big)$ and $S\big( {\left\{ {{{\bm w}_l}} \right\}_{l = 0}^I,{\bm v}} \big)$ are differentiable, their gradients can be obtained. Based on this observation, a simple projected gradient ascent \eqref{eq:17} (shown at the top of next page) is adopted to find an ascent point, where ${1 \mathord{\left/
		{\vphantom {1 {{\alpha ^k}}}} \right.
		\kern-\nulldelimiterspace} {{\alpha ^k}}}$ is an appropriately chosen stepsize.
	
\begin{figure*}
	\begin{equation}\label{eq:17}
		\begin{split}
			\left\{ {\left\{ {{\bm w}_l^{k + 1}} \right\}_{l = 0}^I,{{\bm v}^{k + 1}}} \right\} = \mathop {\arg \min }\limits_{\left\{ {{{\bm w}_l}} \right\}_{l = 0}^I,{\bm v}}& \sum\limits_{l = 0}^I {\left\| {{{\bm w}_l} - \left( {{\bm w}_l^k + \frac{1}{{{\alpha ^k}}}\left[ {{{\left( {{{\bm y}^{k + 1}}} \right)}^T}{\nabla _{{{\bm w}_l}}}{\rm{{\bm \psi}}}\left( {\left\{ {{\bm w}_l^k} \right\}_{l = 0}^I} \right) + {c^{k + 1}}{\nabla _{{{\bm w}_l}}}S\left( {\left\{ {{\bm w}_l^k} \right\}_{l = 0}^I,{{\bm v}^k}} \right)} \right]} \right)} \right\|_2^2} \\
			&+ \left\| {{\bm v} - \left( {{{\bm v}^k} + \frac{{{c^{k + 1}}}}{{{\alpha ^k}}}{\nabla _{\bm v}}S\left( {\left\{ {{\bm w}_l^k} \right\}_{l = 0}^I,{{\bm v}^k}} \right)} \right)} \right\|_2^2\\
			s.t.&\;\;\sum\nolimits_{l = 0}^I {\left\| {{{\bm w}_l}} \right\|_2^2}  \le {P_{BS}},\\
			&\;\;{\bm v} \in {\Omega _{\bm v}}.
		\end{split}
	\end{equation}
\rule[5pt]{18.07cm}{0.1em}
\end{figure*}

According to \cite{beck:17}, in order to ensure \eqref{eq:17} is an ascent point of the original maximization subproblem \eqref{eq:43}, ${{\alpha ^k}}$ in \eqref{eq:17} should be larger than the Lipschitz constant of the objective function \eqref{eq:43}. Since the utility function ${\bm \psi} \big( {\left\{ {{{\bm w}_l}} \right\}_{l = 0}^I} \big)$ and constraint $S\big( {\left\{ {{{\bm w}_l}} \right\}_{l = 0}^I,{\bm v}} \big)$ are assumed to be twice continuously differentiable, the Lipschitz constant of the objective function \eqref{eq:43} exists. However, in practice, the Lipschitz constant may be difficult to find (e.g., when involving an implicit function such as ${{D_i}\big( {\left\{ {{{\bm w}_l}} \right\}_{l = 0}^I} \big)}$). Inspired by the diminishing stepsize rule in the projected gradient method \cite{Bin:06}, ${{\alpha ^k}}$ can be chosen as an increasing sequence with $\mathop {\lim }\limits_{k \to \infty } {\alpha ^k} = \infty $. Then, as the iteration goes, ${{\alpha ^k}}$ will eventually exceed the Lipschitz constant of \eqref{eq:43} and guarantee to provide an ascent point of \eqref{eq:43}. Similar to the regularization parameters in the inner minimization subproblem, theoretical results on the schedule of ${\alpha ^k}$ would be given later.

\begin{algorithm}[t]
	\renewcommand{\algorithmicrequire}{\textbf{Input:}} 
	\renewcommand{\algorithmicensure}{\textbf{General step:}} 
	\caption{Modified AO iteration for solving \eqref{eq:26}}
	\begin{algorithmic}[1]
		\REQUIRE ~~\\ 
		The upper bound of the iteration numbers $K$\\
		\renewcommand{\algorithmicrequire}{\textbf{Output:}}
		\REQUIRE ~~\\
		$\big\{ {\left\{ {{\bm w}_l^{best}} \right\}_{l = 0}^I,{{\bm v}^{best}}} \big\}$\\
		\ENSURE ~~\\
		\STATE Initialize $\big\{ {\left\{ {{\bm w}_l^1} \right\}_{l = 0}^I,{{\bm v}^1},{c^1},{{\bm y}^1}} \big\}$ as any feasible point of the problem \eqref{eq:26}. Set $\left\{ {{\bm w}_l^{best}} \right\}_{l = 0}^I \leftarrow \left\{ {{\bm w}_l^1} \right\}_{l = 0}^I$ and ${{\bm v}^{best}} \leftarrow {{\bm v}^1}$\\
		\STATE \textbf{For} $k = 1,2,...,K$\\
		\STATE \quad Calculate ${\lambda _c^k}$, ${\lambda_{\bm y}^k}$, $\beta ^k$, ${\alpha ^k}$ and ${{C}^k}$\\
		\STATE \quad Obtain $\left\{ {{c^{k + 1}},{{\bm y}^{k + 1}}} \right\}$ by solving \eqref{eq:16}\\
		\STATE \quad Obtain $\big\{ {\left\{ {{\bm w}_l^{k + 1}} \right\}_{l = 0}^I,{{\bm v}^{k + 1}}} \big\}$ by solving \eqref{eq:17}\\
		\STATE \quad \textbf{If} the current iteration solution $\big\{ {\left\{ {{\bm w}_l^{k + 1}} \right\}_{l = 0}^I,{{\bm v}^{k + 1}}} \big\}$\\
		\quad \quad satisfies $S\big( {{\left\{ {{\bm w}_l^{k+1}} \right\}_{l = 0}^I},\bm v^{k+1}} \big) \ge 0$, ${{\bm v}^{k + 1}} \in {\Omega _{\bm v}}$ and \\ \quad \quad $\mathop {\min }\limits_{{\bm y} \in {\Omega _{\bm y}}} {{\bm y}^T}{\bm \psi} \big( {\left\{ {{\bm w}_l^{k + 1}} \right\}_{l = 0}^I} \big) \ge \mathop {\min }\limits_{{\bm y} \in {\Omega _{\bm y}}} {{\bm y}^T}{\bm \psi} \big( {\left\{ {{\bm w}_l^{best}} \right\}_{l = 0}^I} \big)$\\
		\STATE \quad \quad $\left\{ {{\bm w}_l^{best}} \right\}_{l = 0}^I \leftarrow \left\{ {{\bm w}_l^{k+1}} \right\}_{l = 0}^I$ and ${{\bm v}^{best}} \leftarrow {{\bm v}^{k+1}}$\\
		\STATE \quad \textbf{End}
		\STATE \textbf{until} the iteration number $k$ meets the upper bound $K$.\\
	\end{algorithmic}
\end{algorithm}

\subsection{Overall modified AO algorithm}

The overall proposed AO iteration for solving \eqref{eq:26} is summarized in Algorithm 1. Notice that as this modified AO iteration involves two opposing layers, the objective value generated by the sequence $\big\{ {\left\{ {{\bm w}_l^k} \right\}_{l = 0}^I,{{\bm v}^k},{c^k},{{\bm y}^k}} \big\}$ is not monotonic. Therefore, in Algorithm 1, we need steps 6 and 7 to keep track of the current best solution.

For the outer iteration, due to the general forms of ${\rm{{\bm \psi}}}\big( {\left\{ {{\bm w}_l^k} \right\}_{l = 0}^I} \big)$ and $S\big( {\left\{ {{{\bm w}_l}} \right\}_{l = 0}^I,{\bm v}} \big)$, there is no guarantee we could obtain a suboptimal or stationary point of the objective function \eqref{eq:43} even we execute \eqref{eq:17} for many iterations. Moreover, the max-min problem \eqref{eq:26} involves a coupled structure where the inner minimization influences the outer maximization. Updating the outer variables multiple times before updating the inner variables can lead to a misalignment between the two layers, which causes the algorithm to spend unnecessary time adjusting the outer variables based on outdated or inconsistent inner solutions, effectively slowing down the convergence \cite{Jiawei:20}. Therefore, we propose to execute \eqref{eq:17} for only one iteration and then move to the inner minimization of $\left\{ {c,{\bm y}} \right\}$ in \eqref{eq:16}. Although we only obtain an ascent point of \eqref{eq:43} by a single iteration, solution quality of the whole AO iteration can still be obtained. This is because the convergence and solution quality depend on the stepsize and the weight of the introduced regularization terms \cite{Jiefei:24,Zi:23,Maher:19} in \eqref{eq:16}. It does not rely on whether we use multiple projected gradient steps or single projected step in outer maximization layer. The detailed convergence analysis for the whole modified AO iteration will be provided next.

\subsection{Solution quality analysis}
Although there exist some solution quality analyses for algorithms solving non-concave convex max-min problems \cite{Maher:19,Wang:20,Jiawei:20,Zi:23,Jiefei:24,Kaiqing:21}, they are not applicable to the proposed AO iteration for solving \eqref{eq:26}. Firstly, existing analyses are carried out for algorithms with known Lipschitz constant \cite{Maher:19,Zi:23,Jiawei:20,Jiefei:24} or upper bound of the Hessian matrix norm \cite{Wang:20} of the objective function \eqref{eq:26a}. Although the Lipschitz constant or upper bound of the Hessian matrix norm of \eqref{eq:26a} exist since \eqref{eq:26a} is twice continuously differentiable, the general form of function ${\bm \psi} \big( {\left\{ {{{\bm w}_l}} \right\}_{l = 0}^I} \big)$ and sensing metric $S\big( {\left\{ {{{\bm w}_l}} \right\}_{l = 0}^I,{\bm v}} \big)$ in \eqref{eq:26a} make them unobtainable. Furthermore, the solution quality analyses of these existing works require the constraints of the max-min problem to be bounded \cite{Kaiqing:21,Jiawei:20}. Unfortunately, this is not valid in our case since the constraint $c \ge 0$ of the auxiliary variable is obviously unbounded. Therefore, a dedicated analysis is needed for the proposed algorithm.

Before we reveal the solution quality guarantee of the algorithm, we need to define the $\varepsilon $-stationary point which is commonly used in the analysis of non-monotonic algorithms for solving max-min problem.

\begin{myDef} (\textbf{$\varepsilon $-stationary point})\cite[Def 3.1]{Jiawei:20}, \cite[Def 3.1]{Jiefei:24} Given a max-min problem
	$\mathop {\max }\limits_{{\bm x} \in {\rm{{\cal X}}}} \mathop {\min }\limits_{{\bm z} \in {\rm{{\cal Z}}}} f\left( {{\bm x},{\bm z}} \right)$,
	let ${{\mathbb I}_{\rm{{\cal X}}}}\left( {\bm x} \right)$ and ${{\mathbb I}_{\rm{{\cal Z}}}}\left( {\bm z} \right)$ be the indicator functions of the constraints ${{\bm x} \in {\rm{{\cal X}}}}$ and ${{\bm z} \in {\rm{{\cal Z}}}}$, respectively. We call $\left\{ {{{\bm x}^k},{{\bm z}^k}} \right\}$ is a $\varepsilon $-stationary point of this max-min problem if there exists $\bm d \in \partial {{\mathbb I}_{\rm{{\cal X}}}}\left( {\bm x} \right)$ and $\bm e \in \partial {{\mathbb I}_{\rm{{\cal Z}}}}\left( {\bm z} \right)$ such that $\left\| {\left[ \begin{array}{l}
			{\nabla _{\bm x}}f\left( {{\bm x},{\bm z}} \right) - {\bm d}\\
			{\nabla _{\bm z}}f\left( {{\bm x},{\bm z}} \right) + {\bm e}
		\end{array} \right]} \right\|_2$ $\le \varepsilon $, where $\partial $ is the subgradient operation on the non-smooth indicator functions.
\end{myDef}

\noindent Next we present the convergence and solution quality of Algorithm 1 in Proposition 2.

\begin{myPro}
	Define $T\left( \varepsilon  \right)$ as the iteration index of Algorithm 1 obtaining a $\varepsilon$-stationary point of \eqref{eq:26} for the first time. Suppose the sensing metric $S\big( {\left\{ {{{\bm w}_l}} \right\}_{l = 0}^I,{\bm v}} \big)$ and utility function ${\bm \psi} \big( {\left\{ {{{\bm w}_l}} \right\}_{l = 0}^I} \big)$ are twice continuously differentiable and the feasible set ${\Omega _{\bm v}}$ is compact. Also, assume that the global optimal solution of \eqref{eq:16} and \eqref{eq:17} can be obtained. Then, an upper bound of $T\left( \varepsilon  \right)$ is given by $T\left( \varepsilon  \right) \le {\rm{{\cal O}}}\left( {{\varepsilon ^{ - 6}}} \right)$ under the setting ${\lambda_c ^k}= {\lambda_{\bm y} ^k}\propto {k^{ - \frac{1}{4}}}$, ${\beta ^k} \propto {k^{ - 3}}$, ${\alpha ^k} \propto {k^{\frac{1}{3}}}$ and ${C^k} \propto \ln k$.
\end{myPro}
\emph{Proof:} Due to page limitation, we only provide a sketch of the proof in the following. More details can be found in Appendix B in supplementary material. First, we can prove the smoothness of the regularized inner minization subproblem \eqref{eq:16} due to the twice continuously differentiable of $S\big( {\left\{ {{{\bm w}_l}} \right\}_{l = 0}^I,{\bm v}} \big)$, ${\bm \psi} \big( {\left\{ {{{\bm w}_l}} \right\}_{l = 0}^I} \big)$ and gradient consistence theory \cite{Tianyi:20}. Second, we can prove the difference of the optimal objective value of the subproblem \eqref{eq:16} between the iteration points ${\big\{ {\left\{ {{\bm w}_l^{k + 1}} \right\}_{l = 0}^I,{{\bm v}^{k + 1}}} \big\}}$ and ${\big\{ {\left\{ {{\bm w}_l^{k}} \right\}_{l = 0}^I,{{\bm v}^{k}}} \big\}}$ is lower bounded. Lastly, applying the definition of $\varepsilon $-stationary point and convergent Bertrand series \cite{Bertrand:23}, Proposition 2 is proved.$\Box$

\vspace{0.2cm}
According to Proposition 2, the proposed algorithm can find at least one $\varepsilon$-stationary point of \eqref{eq:26} within ${\rm{{\cal O}}}\left( {{\varepsilon ^{ - 6}}} \right)$ iterations (${\rm{{\cal O}}}\left( {{\varepsilon ^{ - 6}}} \right)$ can be interpreted as a conservative convergence rate). This means that as the number of iteration $k$ increases, we are guaranteed to find a $\varepsilon$-stationary point with $\varepsilon$ arbitrarily close to zero, which corresponds to a stationary point of \eqref{eq:26}. According to Proposition 1, a stationary point of \eqref{eq:26} is at least a stationary point of \eqref{eq:19}. Therefore, Algorithm 1 is guaranteed to give at least a stationary point of \eqref{eq:19} as the number of iteration goes to infinity.

Regarding the compact set assumption in Proposition 2, as the wireless resources related to the variable $\bm v$ should not be unlimited, the feasible set ${\Omega _{\bm v}}$ should be compact. On the other hand, Proposition 2 requires the global optimal solution of \eqref{eq:16} and \eqref{eq:17}, which seems to be restrictive since the objective functions in \eqref{eq:16} and \eqref{eq:17} contain the general sensing performance metric ${S\big( {\left\{ {{\bm w}_l} \right\}_{l = 0}^I,{\bm u}} \big)}$. But a closer look reveals that the sensing metrics in \eqref{eq:16} and \eqref{eq:17} only depend on the $k^{th}$ iteration solution but not the unknown variables, making the objective functions of \eqref{eq:16} and \eqref{eq:17} only contain at most quadratic terms of the optimization variables. This makes it possible to obtain the global optimal solutions of \eqref{eq:16} and \eqref{eq:17}, which will be discussed in the next section.
		
\section{Solving \eqref{eq:16} and \eqref{eq:17} under different ISAC sensing metrics}
\subsection{Obtaining the global optimal solution of \eqref{eq:16}}
Notice that $\bm y$ and $c$ are decoupled in both the objective function and the constraints in \eqref{eq:16}. This makes it possible to solve \eqref{eq:16} by independently minimizing subproblems w.r.t. $\bm y$ and $c$. The details are given as follows.

\noindent \underline{\emph{Updating $\bm y$:}} Focusing on the terms related to $\bm y$ in \eqref{eq:16}, the problem becomes
\begin{equation}\label{eq:27}
	\begin{split}
		\mathop {\min }\limits_{{\bm y}}\;\;& {{\bm y}^T}{\bm \psi }\left( {\left\{ {{{\bm w}_l^k}} \right\}_{l=0}^I} \right)+\frac{{\lambda _{\bm y}^k}}{2}{\left\| {\bm y} \right\|^2} \\
		s.t.\;\;&{\bm y} \in {\Omega _{\bm y}}.
	\end{split}
\end{equation}
Since this is a projection problem with a quadratic objective function, the optimal solution is
\begin{equation}\label{eq:50}
	{y^{k + 1}} = {{\rm Proj}_{{\Omega _{\bm y}}}}\left( { - \frac{{{\bm \psi} \left( {\left\{ {{\bm w}_l^k} \right\}_{l = 0}^I} \right)}}{{\lambda _{\bm y}^k}}} \right),
\end{equation}
where ${{\rm Proj}_{\rm{{\cal A}}}}\left( b \right)$ is to project $b$ to the feasible set ${\rm{{\cal A}}}$. If the constraint on ${\bm y}$ is a singleton, $\bm y$ is a constant and \eqref{eq:27} can be skipped. In the worst utility maximization, ${\Omega _{\bm y}}: = \big\{ {\sum\nolimits_{i = 1}^I {{y_i}}  = 1,{y_i} \ge 0,\forall i} \big\}$, and \eqref{eq:50} under such constraint can be obtained as \cite{Zongze:21}
\begin{equation}\label{eq:40}
	{{\bm y}^{k + 1}}\left( \iota \right) = {\left[ { - \frac{{{\bm \psi} \big( {\big\{ {{\bm w}_l^k} \big\}_{l = 0}^I} \big) + \iota }}{{\lambda _{\bm y}^k}}} \right]^ + },
\end{equation}
where $\iota$ satisfies ${{\bm 1}^T}{{\bm y}^{k + 1}}\left( \iota \right) = 1$ and can be found by the bisection method from the interval $\big[ {\mathop {\min }\limits_i \big\{ { -{\psi _i}\big( {\left\{ {{\bm w}_l^k} \right\}_{l = 0}^I} \big)} \big\}} \big. - {\lambda _{\bm y}^k},$ $\big. {\mathop {\max }\limits_i \big\{ { - {\psi _i}\big( {\left\{ {{\bm w}_l^k} \right\}_{l = 0}^I} \big)} \big\}} \big]$.

\noindent \underline{\emph{Updating $c$:}} Focusing on the terms related to $c$ in \eqref{eq:16}, the problem becomes
\begin{equation}\label{eq:28}
	\begin{split}
		\mathop {\min }\limits_c&\;\; cS\left( {{\left\{ {{\bm w}_l^k} \right\}_{l = 0}^I},{\bm v^k}} \right) + \frac{{{\lambda _c^k}}}{2}{c^2} + \frac{{{\beta ^k}}}{2}{\left( {c - {c^k}} \right)^2}\\
		s.t.&\;\;c \in \left[ {0,{C_k}} \right].
	\end{split}
\end{equation}
The optimization problem \eqref{eq:28} can be solved in closed-form with a box constraint projection
\begin{equation}\label{eq:41}
	c^{k+1} = {{\rm Proj}_{\left[ {0,{C_k}} \right]}}\left( {\frac{{{\beta ^k}{c^k} - S\left( {{{\left\{ {{\bm w}_l^k} \right\}_{l = 0}^I}},{{\bm v}^k}} \right)}}{{{\lambda _c^k} + {\beta ^k}}}} \right).
\end{equation}
Since the updating of $\bm y$ and $c$ involve simple operations, \eqref{eq:16} can be solved very efficiently.

\subsection{Obtaining the global optimal solution of \eqref{eq:17}}
Notice that $\left\{ {{{\bm w}_l}} \right\}_{l = 0}^I$ and $\bm v$ are decoupled in both objective function and constraints of \eqref{eq:17}.Therefore, ${\left\{ {{\bm w}_l^{k + 1}} \right\}_{l = 0}^I}$ and ${{\bm v}^{k + 1}}$ can be separately solved without losing the optimality. 

\noindent \underline{\emph{Updating ${\left\{ {{{\bm w}_l}} \right\}_{l = 0}^I}$:}} Focusing on the terms related to ${\left\{ {{{\bm w}_l}} \right\}_{l = 0}^I}$ in \eqref{eq:17}, the maximization problem is given by
\begin{equation}\label{eq:29}
\begin{split}
	&\mathop {\min }\limits_{\left\{ {{{\bm w}_l}} \right\}_{l = 0}^I} \sum\limits_{l = 0}^I {\left\| {{{\bm w}_l} - \left( {{\bm w}_l^k + \frac{1}{{{\alpha ^k}}}{{\left( {{{\bm y}^{k + 1}}} \right)}^T}{\nabla _{{{\bm w}_l}}}{\bm \psi} \left( {\left\{ {{\bm w}_l^k} \right\}_{l = 0}^I} \right)} \right.} \right.} \\
	&\;\;\;\;\;\;\;\;\;\;\;\left. {\left. { + \frac{{{c^{k + 1}}}}{{{\alpha ^k}}}{\nabla _{{{\bm w}_l}}}S\left( {\left\{ {{\bm w}_l^k} \right\}_{l = 0}^I,{{\bm v}^k}} \right)} \right)} \right\|_2^2\\
	&\;\;\;s.t.\;\;\;\sum\nolimits_{l = 0}^I {\left\| {{{\bm w}_l}} \right\|_2^2 \le {P_{BS}}} .
\end{split}
\end{equation}
Notice that the constraint of \eqref{eq:29} is a sum power constraint, its closed-form solution is given by
\begin{equation}\label{eq:31}
	\begin{split}
		{\bm w}_l^{k + 1}\left( \chi  \right) =& \frac{{{\alpha ^k}{\bm w}_l^k + {{\left( {{{\bm y}^{k + 1}}} \right)}^T}{\nabla _{{{\bm w}_l}}}{\bm \psi} \left( {\left\{ {{\bm w}_l^k} \right\}_{l = 0}^I} \right)}}{{{\alpha ^k} + 2\chi }}\\
		&+ \frac{{{c^{k + 1}}{\nabla _{{{\bm w}_l}}}S\left( {\left\{ {{\bm w}_l^k} \right\}_{l = 0}^I,{{\bm v}^k}} \right)}}{{{\alpha ^k} + 2\chi }},
	\end{split}
\end{equation}
where $\chi  = 0$ if $\sum\nolimits_{l = 0}^I {\left\| {{\bm w}_l^{k + 1}\left( 0 \right)} \right\|_2^2}  \le {P_{BS}}$. Otherwise, $\chi $ satisfies the equation $\sum\nolimits_{l = 0}^I {\left\| {{\bm w}_l^{k + 1}\left( \chi  \right)} \right\|_2^2}  = {P_{BS}}$, which can be found by the bisection method from the interval $\big[ {0,\sqrt {\frac{1}{{4{P_{BS}}}}\sum\nolimits_{l = 0}^I {{{\big\| {{\alpha ^k}{\bm w}_l^{k + 1}\left( 0 \right)} \big\|}^2}} } } \big]$.

\noindent \underline{\emph{Updating ${\bm v}$:}} From \eqref{eq:25}, the optimization problem w.r.t. $\bm v$ is given by 

\begin{equation}\label{eq:34}
	\begin{split}
		\mathop {\min }\limits_{\bm v}&\;\; {\big\| {{\bm v} - \big( {{{\bm v}^k} + {{{c^{k + 1}}{\nabla _{\bm v}}S\big( {\big\{ {{\bm w}_l^k} \big\}_{l = 0}^I,{{\bm v}^k}} \big)} \mathord{\big/
							{\vphantom {{{c^{k + 1}}{\nabla _{\bm v}}S\big( {\big\{ {{\bm w}_l^k} \big\}_{l = 0}^I,{{\bm v}^k}} \big)} {{\alpha ^k}}}} \big.
							\kern-\nulldelimiterspace} {{\alpha ^k}}}} \big)} \big\|^2}\\
		s.t.&\;\;{\bm v} \in {\Omega _{\bm v}}.
	\end{split}
\end{equation}
This is a simple projection problem and the optimal solution is 
\begin{equation}\label{eq:42}
	{{\bm v}^{k + 1}} = {{\rm Proj}_{{\Omega _{\bm v}}}\big( {{{\bm v}^k} + {{{c^{k + 1}}{\nabla _{\bm v}}S\big( {\big\{ {{\bm w}_l^k} \big\}_{l = 0}^I,{{\bm v}^k}} \big)} \mathord{\big/
					{\vphantom {{{c^{k + 1}}{\nabla _{\bm v}}S\big( {\big\{ {{\bm w}_l^k} \big\}_{l = 0}^I,{{\bm v}^k}} \big)} {{\alpha ^k}}}} \big.
					\kern-\nulldelimiterspace} {{\alpha ^k}}}} \big)}.
\end{equation}
Notice that \eqref{eq:31} and \eqref{eq:42} only contain simple one dimensional bisection search and projection update, thus solving \eqref{eq:17} is very efficient.

\subsection{Computation of gradients in various ISAC scenarios}
Notice that computation of various gradients only occurs when solving \eqref{eq:17}. As an example, we compute the gradient ${{\nabla _{{{\bm w}_l}}}{\bm \psi} \left( {\left\{ {{\bm w}_l^k} \right\}_{l = 0}^I} \right)}$ when ${\psi _i}\big(  \cdot  \big)$ is the secrecy rate ${{{\big\{ {{{\rm{{\cal R}}}_i}\big( {\left\{ {{{\bm w}_l}} \right\}_{l = 0}^I} \big) - {D_i}\big( {\left\{ {{{\bm w}_l}} \right\}_{l = 0}^I} \big)} \big\}}_{\forall i}}}$:
\begin{equation}\label{eq:39}
	\begin{split}
		&{\nabla _{{{\bm w}_l}}}{\bm \psi} \left( {\left\{ {{{\bm w}_l}} \right\}_{l = 0}^I} \right)\\
		&= \left[ {{\nabla _{{{\bm w}_l}}}{{\rm{{\cal R}}}_1}\left( {\left\{ {{{\bm w}_l}} \right\}_{l = 0}^I} \right) - {\nabla _{{{\bm w}_l}}}{D_1}\left( {\left\{ {{{\bm w}_l}} \right\}_{l = 0}^I} \right), \ldots ,} \right.\\
		&\;\;\;\;\;\;\left. {{\nabla _{{{\bm w}_l}}}{{\rm{{\cal R}}}_I}\left( {\left\{ {{{\bm w}_l}} \right\}_{l = 0}^I} \right) - {\nabla _{{{\bm w}_l}}}{D_I}\left( {\left\{ {{{\bm w}_l}} \right\}_{l = 0}^I} \right)} \right],
	\end{split}
\end{equation}
where ${\nabla _{{{\bm w}_l}}}{{{\rm{{\cal R}}}_i}\big( {\left\{ {{{\bm w}_l}} \right\}_{l = 0}^I} \big)}$ and ${\nabla _{{{\bm w}_l}}}{{D_i}\big( {\left\{ {{{\bm w}_l}} \right\}_{l = 0}^I} \big)}$ are respectively given in \eqref{eq:35} (shown at the top of this page), and ${\tau  _{i,t}} = \left\| {{{\bm w}_i}} \right\|_2^2 + \big( {{2^{{D_i}\left( {\left\{ {{{\bm w}_l}} \right\}_{l = 0}^I} \right)}} - 1} \big)\left\| {{{\bm w}_t}} \right\|_2^2$. If there is no eavesdropper, we can just set ${D_i}\big( {\left\{ {{{\bm w}_l}} \right\}_{l = 0}^I} \big) \equiv 0$.
\begin{figure*}
	\begin{subequations}\label{eq:35}
		\begin{align}
			&{\nabla _{{{\bm w}_l}}}{{\rm{{\cal R}}}_i}\left( {\left\{ {{{\bm w}_l}} \right\}_{l = 0}^I} \right) = \left\{ \begin{array}{l}
				\frac{{2{{\bar {\bm h}}_i}{\bm h}_i^T{{\bm w}_l}}}{{\left( {\sum\nolimits_{t = 0}^I {{{\left| {{\bm h}_i^T{{\bm w}_t}} \right|}^2}}  + \sigma _i^2} \right)\ln 2}} - \frac{{2{{\bar {\bm h}}_i}{\bm h}_i^T{{\bm w}_l}}}{{\left( {\sum\nolimits_{t = 0,t \ne i}^I {{{\left| {{\bm h}_i^T{{\bm w}_t}} \right|}^2}}  + \sigma _i^2} \right)\ln 2}},\;\;l \ne i,\\
				{{2{{\bar {\bm h}}_i}{\bm h}_i^T{{\bm w}_l}} \mathord{\left/
						{\vphantom {{2{{\bar {\bm h}}_i}{\bm h}_i^T{{\bm w}_l}} {\left[ {\left( {\sum\nolimits_{t = 0}^I {{{\left| {{\bm h}_i^T{{\bm w}_t}} \right|}^2} + \sigma _i^2} } \right)\ln 2} \right]}}} \right.
						\kern-\nulldelimiterspace} {\left[ {\left( {\sum\nolimits_{t = 0}^I {{{\left| {{\bm h}_i^T{{\bm w}_t}} \right|}^2} + \sigma _i^2} } \right)\ln 2} \right]}},\;\;\;\;\;\;\;\;\;\;\;\;\;\;\;\;l = i.
			\end{array} \right.\\
			&{\nabla _{{{\bm w}_l}}}{D_i}\left( {\left\{ {{{\bm w}_l}} \right\}_{l = 0}^I} \right) = \left\{ \begin{array}{l}
				\frac{{\left( {1 - {2^{{D_i}\left( {\left\{ {{{\bm w}_l}} \right\}_{l = 0}^I} \right)}}} \right)2{{\bm w}_l}}}{{{2^{{D_i}\left( {\left\{ {{{\bm w}_l}} \right\}_{l = 0}^I} \right)}}{\tau _{i,l}}\left[ {\frac{{\ln 2}}{{\mathop {\max }\limits_j \left\{ {{{\rho _j^2} \mathord{\left/
												{\vphantom {{\rho _j^2} {\varsigma _j^2}}} \right.
												\kern-\nulldelimiterspace} {\varsigma _j^2}}} \right\}\left\| {{{\bm w}_i}} \right\|_2^2}} + \sum\limits_{t = 0,t \ne i}^I {\frac{{\left\| {{{\bm w}_t}} \right\|_2^2\ln 2}}{{{\tau _{i,t}}}}} } \right]}},\;\;\;l \ne i,\\
			{{\left( {{2^{{D_i}\left( {\left\{ {{{\bm w}_l}} \right\}_{l = 0}^I} \right)}} - 1} \right)2{{\bm w}_i}} \mathord{\left/
					{\vphantom {{\left( {{2^{{D_i}\left( {\left\{ {{{\bm w}_l}} \right\}_{l = 0}^I} \right)}} - 1} \right)2{{\bm w}_i}} {\left[ {\left\| {{{\bm w}_i}} \right\|_2^2{2^{{D_i}\left( {\left\{ {{{\bm w}_l}} \right\}_{l = 0}^I} \right)}}\ln 2} \right]}}} \right.
					\kern-\nulldelimiterspace} {\left[ {\left\| {{{\bm w}_i}} \right\|_2^2{2^{{D_i}\left( {\left\{ {{{\bm w}_l}} \right\}_{l = 0}^I} \right)}}\ln 2} \right]}},l = i.
			\end{array} \right.
		\end{align}
	\end{subequations}
\rule[0pt]{18.07cm}{0.1em}
\end{figure*}

On the other hand, we need ${\big\{ {{\nabla _{{{\bm w}_l}}}S\big( {{\left\{ {{{\bm w}_l}} \right\}_{l = 0}^I},{\bm v}} \big)} \big\}_{\forall l}}$ and ${{\nabla _{\bm v}}S\big( {{\left\{ {{{\bm w}_l}} \right\}_{l = 0}^I},{\bm v}} \big)}$ in \eqref{eq:31} and \eqref{eq:42}, respectively. Below, we list the gradients of three sensing performance metrics mentioned in Section III.

\begin{enumerate}[leftmargin=0.5cm]
	\item \underline{\emph{Sensing SINR:}} In this case, ${\bm v}={\bm u}$, and $S\left( {\left\{ {{{\bm w}_l}} \right\}_{l = 0}^I,{\bm v}} \right)$ is given by the left hand side of \eqref{eq:13a}. Correspondingly, the gradients of ${S\big( {{\left\{ {{{\bm w}_l}} \right\}_{l = 0}^I},{\bm v}} \big)}$ are given as
	\begin{subequations}\label{eq:36}
		\begin{align}
			&{\nabla _{{{\bm w}_l}}}S\left( {{\left\{ {{{\bm w}_l}} \right\}_{l = 0}^I},{\bm v}} \right) = 2{\bm A}_0^H\bar {\bm v}{{\bm v}^T}{{\bm A}_0}{{\bm w}_l} \nonumber \\
			&\;\;\;\;\;\;\;\;\;\;\;\;\;\;\;\;\;\;\;\;\;\;\;\;\;\;\;\;\;\;\;\;\;\;- 2\Gamma \sum\nolimits_{m = 1}^M {{\bm A}_m^H\bar {\bm v}{{\bm v}^T}{{\bm A}_m}{{\bm w}_l}}.\label{eq:36a}\\
			&{\nabla _{\bm u}}S\left( {{\left\{ {{{\bm w}_l}} \right\}_{l = 0}^I},{\bm v}} \right) = 2\sum\nolimits_{l = 0}^I {{\bar {\bm A}}_0}{{\bar {\bm w}}_l}{\bm w}_l^T{\bm A}_0^T{\bm v} \nonumber\\
			&\;\;\;\;\;\;\;\;\;\;\;\;\;\;\;\;\;\;\;\;\;\;\;\;\;\;\;\;\;- 2\Gamma \sum\nolimits_{l = 0}^I {\sum\nolimits_{m = 1}^M {{{\bar {\bm A}}_m}{{\bar {\bm w}}_l}{\bm w}_l^T{\bm A}_m^T{\bm v}} }  .\label{eq:36b}
		\end{align}
	\end{subequations}
	Then, the update of $\left\{ {{{\bm w}_l^{k+1}}} \right\}_{l = 0}^I$ can be obtained by putting \eqref{eq:36a} back to \eqref{eq:31}. For the update of $\bm v$, since ${\Omega _{\bm v}}$ is a normalization constraint, the optimal solution $\bm v^{k+1}$ is given by ${{\bm v}^{k + 1}} = \frac{{{\alpha ^k}{{\bm v}^k} + {c^{k + 1}}{\nabla _{\bm v}}S\left( {{\left\{ {{\bm w}_l^k} \right\}_{l = 0}^I},{{\bm v}^k}} \right)}}{{{{\left\| {{\alpha ^k}{{\bm v}^k} + {c^{k + 1}}{\nabla _{\bm v}}S\left( {{\left\{ {{\bm w}_l^k} \right\}_{l = 0}^I},{{\bm v}^k}} \right)} \right\|}_2}}}$.
	
	\item \underline{\emph{Beampattern matching:}} In this case, ${\Omega _{\bm v}} = \emptyset $, and $S\left( {\left\{ {{{\bm w}_l}} \right\}_{l = 0}^I} \right) = \gamma  - MSE\left( {\left\{ {{{\bm w}_l}} \right\}_{l = 0}^I} \right)$, where $MSE\left( {\left\{ {{{\bm w}_l}} \right\}_{l = 0}^I} \right)$ is given by \eqref{eq:20}. Then, the gradient of $S\big( {\left\{ {{{\bm w}_l}} \right\}_{l = 0}^I} \big)$ is given as
	\begin{equation}\label{eq:37}
		\begin{split}
			&{\nabla _{{{\bm w}_l}}}S\left( {\left\{ {{{\bm w}_l}} \right\}_{l = 0}^I} \right) =\frac{4}{T}\sum\nolimits_{t = 1}^T\\
			& {\big[ {{P_d}\left( {{\theta _t}} \right) - {{\bm a}^H}\left( {{\theta _t}} \right)\big( {\textstyle\sum\nolimits_{i = 0}^I {{{\bm w}_i}{\bm w}_i^H} } \big){\bm a}\left( {{\theta _t}} \right)} \big]} {\bm a}\left( {{\theta _t}} \right){{\bm a}^H}\left( {{\theta _t}} \right){{\bm w}_l}.
		\end{split}
	\end{equation}
	The update of $\left\{ {{{\bm w}_l^{k+1}}} \right\}_{l = 0}^I$ is obtained by putting \eqref{eq:37} to \eqref{eq:31}. Since $\bm v$ does not appear in the sensing metric, there is no optimization step for $\bm v$.
	
	\item \underline{\emph{Detection probability:}} In this case, ${\Omega _{\bm v}} = \emptyset $, and $S\left( {\left\{ {{{\bm w}_l}} \right\}_{l = 0}^I} \right)$ is given by the left hand side of \eqref{eq:24}, which gives the gradient of $S\big( {\left\{ {{{\bm w}_l}} \right\}_{l = 0}^I} \big)$ as
	\begin{equation}\label{eq:38}
		{\nabla _{{{\bm w}_l}}}S\left( {\left\{ {{{\bm w}_l}} \right\}_{l = 0}^I} \right) = 2{{\bm A}_0}{{\bm w}_l}.
	\end{equation}
Similar to the above two cases, the update of $\left\{ {{{\bm w}_l^{k+1}}} \right\}_{l = 0}^I$ is obtained by putting \eqref{eq:38} to \eqref{eq:31}. Similar to beampattern matching case, there is no $\bm v$ in the sensing metric, the corresponding optimization step can be skipped.
\end{enumerate}

Notice that the proposed algorithm can handle sensing metrics beyond the above cases. When a new sensing metric is employed, the only change involved is the gradients ${\nabla _{{{\bm w}_l}}}S\big( {\left\{ {{\bm w}_l^k} \right\}_{l = 0}^I,{{\bm v}^k}} \big)$ and ${\nabla _{\bm v}}S\big( {\left\{ {{\bm w}_l^k} \right\}_{l = 0}^I,{{\bm v}^k}} \big)$, while other parts of the algorithm remain the same. For example, if the sensing metric is the mutual information (MI) \cite{xu2025mutual}, we have $S\big( {\big\{ {{{\bm w}_l}} \big\}_{l = 0}^I} \big) = {\log _2}\big( {1 + \frac{1}{{\sigma _{BS}^2}}\sum\nolimits_{l = 0}^I {{\bm w}_l^H\big( {\sum\nolimits_{m = 0}^M {{\bm A}_m^H{{\bm A}_m}} } \big){{\bm w}_l}} } \big) - \delta$, where $\delta$ is the minimal MI needed for sensing, and there is no $\bm v$. To connect with the proposed framework, we only need to compute the gradient of $S\big( {\big\{ {{{\bm w}_l}} \big\}_{l = 0}^I} \big)$ and replace ${\nabla _{{{\bm w}_l}}}S\left( {\left\{ {{\bm w}_l^k} \right\}_{l = 0}^I,{{\bm v}^k}} \right)$ in \eqref{eq:31} with the gradient of the MI metric. This advantage makes the proposed algorithm widely applicable to different ISAC scenarios.

Overall, the complexity of Algorithm 1 is summarized in Table \ref{tab1}, where ${\xi _D}$, ${\xi _{\bm y}}$ and ${\xi _{\bm w}}$ are the accuracy of the bisection search in redundancy rate $\left\{ {{D_i}} \right\}_{i = 1}^I$, auxiliary vector $\bm y$ and beamformers $\left\{ {{{\bm w}_l}} \right\}_{l = 0}^I$, respectively. For the complexity of calculating the gradient of $S\big( {\left\{ {{{\bm w}_l}} \right\}_{l = 0}^I,\bm v} \big)$ with respect to $\left\{ {{{\bm w}_l}} \right\}_{l = 0}^I$ in a single iteration, it is ${{\cal{O}}}\left( {MNI} \right)$, ${{\cal{O}}}\left( {TNI} \right)$ and ${{\cal{O}}}\left( {NI} \right)$ for the sensing SINR, beampattern matching and detection probability, respectively. ${{{\cal{O}}}_{\bm v}}$ is the complexity of updating $\bm v$ in a single iteration, which involves calculating the gradient of $S\big( {\left\{ {{{\bm w}_l}} \right\}_{l = 0}^I,\bm v} \big)$ with respect to $\bm v$ and the projection operation in (34). As $\bm v$ only appears in the sensing SINR case, ${{{\cal{O}}}_v} = {{\cal{O}}}\left( {MNI} \right)$.

\begin{table*}
	\centering
	\begin{threeparttable}[b]
		\caption{Computational complexity of Algorithm 1}
		\label{tab1}
		\begin{tabular}{|c|c|}
			\hline
			&\\[-7pt]
			&Computational complexity\\
			\hline
			&\\[-7.5pt]
			Algorithm without gradient update of $S\left( {\left\{ {{{\bm w}_l}} \right\}_{l = 0}^I,{\bm v}} \right)$ &${{\cal{O}}}\left( {K{I^2}N + K{I^2}\log \left( {\xi _D^{ - 1}} \right) + KI\log \left( {\xi _{\bm y}^{ - 1}} \right) + KNI\log \left( {\xi _{\bm w}^{ - 1}} \right)} \right)$\\
			\hline
			&\\[-6pt]
			Gradient update under sensing SINR&$K \cdot {{\cal{O}}}\left( {MNI} \right) + K \cdot {{{\cal{O}}}_{\bm v}}$\\
			\hline
			&\\[-6pt]
			Gradient update under beampattern matching&$K \cdot {{\cal{O}}}\left( {TNI} \right)$\\
			\hline
			&\\[-6pt]
			Gradient update under detection probability&$K \cdot {{\cal{O}}}\left( {NI} \right)$\\
			\hline
		\end{tabular}
	\end{threeparttable}
\end{table*}

\section{Simulation Results and Discussions}
\begin{figure*}[t]
	\centering
	\subfloat[]{\includegraphics[width=6.03cm]{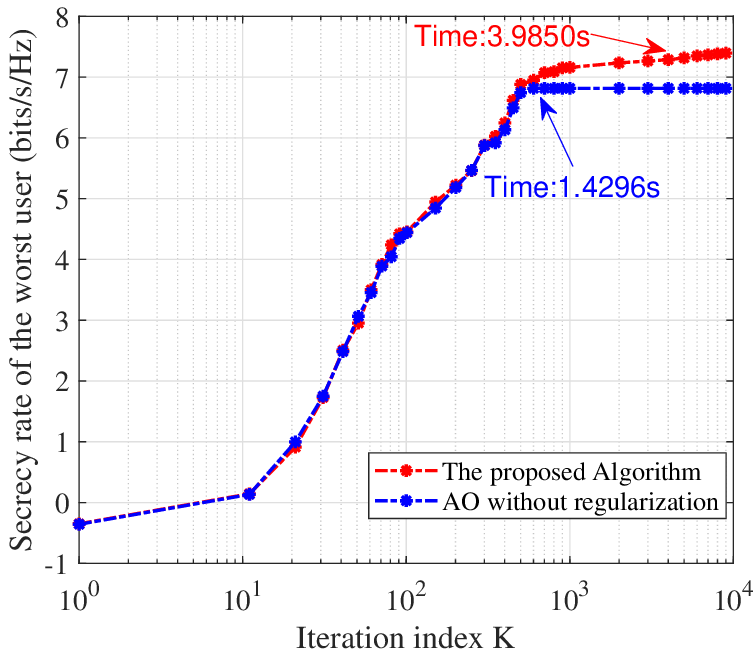}}\hfil
	\subfloat[]{\includegraphics[width=6.03cm]{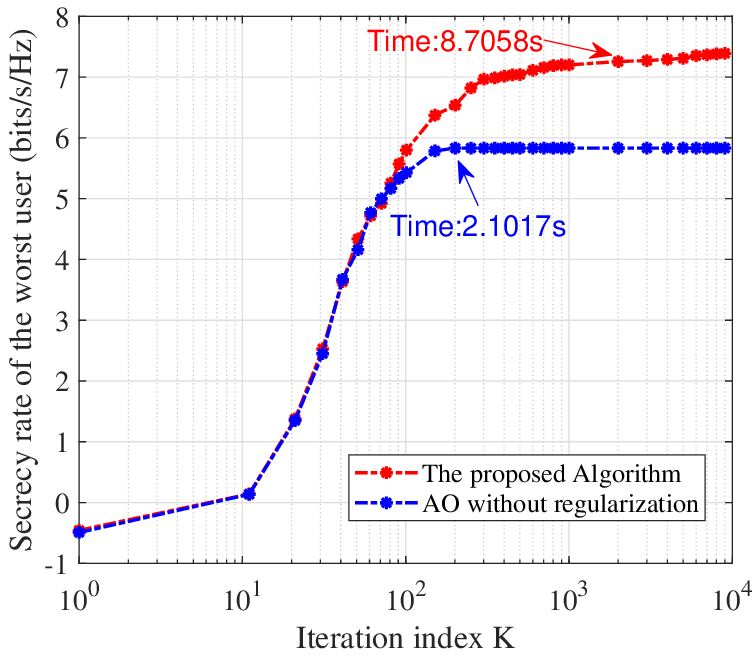}}\hfil 
	\subfloat[]{\includegraphics[width=6.03cm]{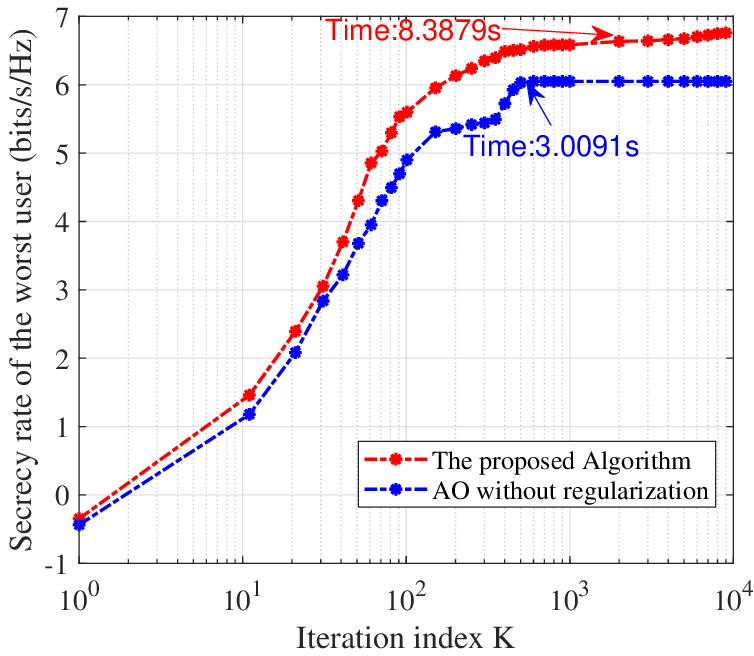}} 	
	\caption{Worst user secrecy rate performance with $I=5$, $J=8$, $N=20$, $M=5$, $P_{BS} = 20{\rm dBm}$. (a) $\Gamma  = 10{\rm dB}$ under sensing SINR constraint, (b) $\gamma  = 0.1$ under beampattern matching constraint, (c) ${P_{FA}} = 0.1$ and $\phi  = 0.9$ under detection probability constraint}
	\label{fig_three}
\end{figure*}

In this section, we evaluate the performance of the proposed algorithm through simulations, with each point in the figures obtained via averaging over 100 simulation trials. In the simulations, the position of the legitimate users, eavesdroppers, scatterings and BS are randomly distributed in a circular area with a radius of $100$m. The carrier frequency is 3.5GHz, the bandwidth is set to $10$ MHz and the noise power spectral density for legitimate users, eavesdroppers and BS are $\sigma _i^2 = \varsigma _j^2 = \sigma _{BS}^2 =  - 96{\rm dBm}/{\rm Hz}$ \cite{Zongze:21}. The target secrecy outage probability level ${\eta  _i} = 0.1$. For the proposed algorithm, the regularization parameters at the $k^{th}$ iteration are $\lambda _c^k = \lambda _{\bm y}^k = {k^{ - \frac{1}{4}}}$, ${\beta ^k} = {k^{ - 3}}$, ${\alpha ^k} = {k^{\frac{1}{3}}}$ and ${C^k} = 10 + \ln k$, which are set according to Proposition 2. The variance of the statistical CSI of eavesdroppers ${\left\{ {{\rho _{j}}} \right\}_{\forall j}}$ is set according to \cite{Zongze:21}, which follows the signal propagation model. The round-trip channel coefficients ${\left\{ {{\beta _m}} \right\}_{\forall m}}$ of the scatterings are set as ${\beta _m}\left( {{\rm dB}} \right) = 20 + 20{\log _{10}}\left( {{d_m}} \right)$ \cite{Silei:25}, where ${{d_m}}$ is the Euclidean distances between the BS and the $m^{th}$ scattering. The number of sampling angles $T=181$ in beampattern matching case are uniformly distributed over the range from $ - {90^ \circ }$ to $ {90^ \circ }$ \cite{Siyi:25}, and the desired beampattern ${P_d}\left( {{\theta _t}} \right)$ follows that from \cite[eq. 11]{Peng:22}. The channels ${\left\{ {{{\bm h}_{i}}} \right\}_{\forall i}}$ of the legitimate users are modeled as Rician fading, which contain both the line-of-sight (LoS) and non-LoS (NLoS) components~\cite{Tse:05}, and is given by ${{\bm h}_i} = \sqrt {{{{L_0}{{\left( {{{{d_{{{\bm h}_i}}}} \mathord{\left/
									{\vphantom {{{d_{{{\bm h}_i}}}} {{d_0}}}} \right.
									\kern-\nulldelimiterspace} {{d_0}}}} \right)}^{ - {\alpha _{\bm h}}}}} \mathord{\left/
				{\vphantom {{{L_0}{{\left( {{{{d_{{{\bm h}_i}}}} \mathord{\left/
												{\vphantom {{{d_{{{\bm h}_i}}}} {{d_0}}}} \right.
												\kern-\nulldelimiterspace} {{d_0}}}} \right)}^{ - {\alpha _{\bm h}}}}} {\left( {{\kappa _{\bm h}} + 1} \right)}}} \right.
				\kern-\nulldelimiterspace} {\left( {{\kappa _{\bm h}} + 1} \right)}}} \left( {\sqrt {{\kappa _{\bm h}}} {\bm h}_i^{LoS} + {\bm h}_i^{NLoS}} \right)$, with each parameter explained below.
\begin{enumerate}[leftmargin=0.5cm]
	\item ${L_0}{\left( {{{{{ d}_{{{\bm h}_{k,i}}}}} \mathord{\left/
						{\vphantom {{{{ d}_{{{\bm h}_{i}}}}} {{{ d}_0}}}} \right.
						\kern-\nulldelimiterspace} {{{ d}_0}}}} \right)^{ - {\alpha _{\bm h}}}}$ is the distance dependent pathloss from the BS to the $i^{th}$ legitimate user, where ${L_0} =  - 30$dB denotes the pathloss at the reference distance ${{d}_0} = 1$m. The parameter ${{{ d}_{{{\bm h}_{i}}}}}$ is the distance between the BS and the $i^{th}$ legitimate user, ${{\alpha _{\bm h}}} = 2$ denotes the pathloss exponent for the BS-user link, and ${{\kappa _{\bm h}}}=5$ is the Rician factor of the BS-user link.
	\item The LoS component ${\bm h}_{i}^{LoS}$ is modeled as the response vector of a uniform linear array. Hence the $n^{th}$ element ${\left[ {{\bm h}_i^{LoS}} \right]_n} = {e^{j2\pi \left( {n - 1} \right)\zeta \sin \left( \omega  \right)}}$, where $\omega $ denotes the angle-of-arrival (AoA) or
		angle-of-departure (AoD) of the array, which is uniformly distributed in $\left[ {0,2\pi } \right)$ and $\zeta=1/2$. On the other hand, ${{\bm h}_{i}^{NLoS}}$ denotes the NLoS channel component with each element obeying the normalized
		complex Gaussian distribution.
\end{enumerate}

First, we show the convergence of the proposed method for worst secrecy rate maximization under three different sensing criteria in Fig. \ref{fig_three}. Here, the times in these figures are counted at the iterations where the arrows point to. Furthermore, to demonstrate the importance of regularizing the AO minimization subproblem \eqref{eq:33}, we also include the simulation of directly executing \eqref{eq:33} without regularization. As can be seen from Fig. \ref{fig_three}, the proposed AO iteration converges fast, with the secrecy rate at $K=1000$ iterations reaches at least 95\% of that at $K=10000$ iterations. Furthermore, the performance degradation resulting from not regularizing \eqref{eq:33} under the criteria of sensing SINR, beampattern matching and detection probability are 11.51\%, 29.73\% and 15.47\%, respectively. This degradation is due to possibly unbounded solution of \eqref{eq:33}, which hinders the subsequent iterations and leads to premature termination of the process. Notice that the usual termination condition, such as the improvement of the objective function between two adjacent iterations being less than a certain tolerance (e.g., ${10^{ - 4}}$), does not apply to the proposed algorithm since it involves non-monotonic iterations. Actually, proper selection of the number of iterations is still an open problem in non-monotonic iterative algorithms \cite{Jiefei:24,Jiawei:20}, which involves the trade-off between the computational time and the performance. In the subsequent simulations, we take results at $K=1000$ and $K=10000$ iterations to demonstrate the performance of the proposed algorithm.

Next, we compare the worst user secrecy rate performance of the proposed algorithm with other optimization methods. The compared approaches are labeled in the form of $A+B$, where $A$ represents the way of handling the inner minimization in \eqref{eq:8a} while $B$ denotes the method for handling the sensing constraint and SOP constraint after applying the Lemma 1. In particular, for handling the inner minimization, we consider either penalty reformulation or auxiliary variable substitution. In penalty reformulation, it uses the log-sum-exp function with an additional penalty weight to approximate the inner minimization problem \cite{Wei-Chiang:15}. On the other hand, in auxiliary variable substitution, the inner minimization problem is replaced with an auxiliary variable, and moved as a constraint~\cite{Ziang:24}. For handling the resultant transformed problem, we consider SDR, projected gradient ascent or SCA. However, notice that the projected gradient ascent cannot be applied to the transformed problem obtained by auxiliary variable substitution since the transformed objective function does not depend on the beamformers $\left\{ {{{\bm w}_l}} \right\}_{l = 0}^I$ and the vector $\bm v$ in the sensing constraint anymore, and thus its derivative with respect to these variables are all zeros. This results in a total of 5 methods being compared. The running times of each algorithm in Figs. \ref{fig_BS} - \ref{fig_Eves} are averaged over all simulated settings in each figure. 

As can be seen in Figs. \ref{fig_BS}(a) and \ref{fig_Eves}(a), which shows the secrecy rate of the worst user under sensing SINR constraint, the proposed algorithm achieves the best secrecy rate. This is due to the following two reasons. Firstly, through the introduced dual variable $c$, the sensing performance constraint is put as part of the objective function. This means that at early stages of iteration, the sensing performance constraint can be violated, which enlarges the feasible set and increases the flexibility of finding a better iteration trajectory. As the iteration goes, the dual variable $c$ would gradually pull the iteration points to satisfy the sensing performance constraint. In this way, the sensing performance constraint is enforced by the proposed algorithm as a soft margin constraint with the margin increasingly tightened. This is in great contrast to other algorithms that treat the sensing performance constraint as hard constraint in all iterations. Secondly, since we alternatively optimize the max layer and min layer, the iteration process is non-monotonic. This enables the proposed algorithm to cross the valley of the bumpy surface to find a better solution. However, other algorithms are monotonic and they cannot go over a valley to find a solution with higher worst case secrecy rate.

\begin{figure*}[t]
	\centering
	\subfloat[]{\includegraphics[width=6.03cm]{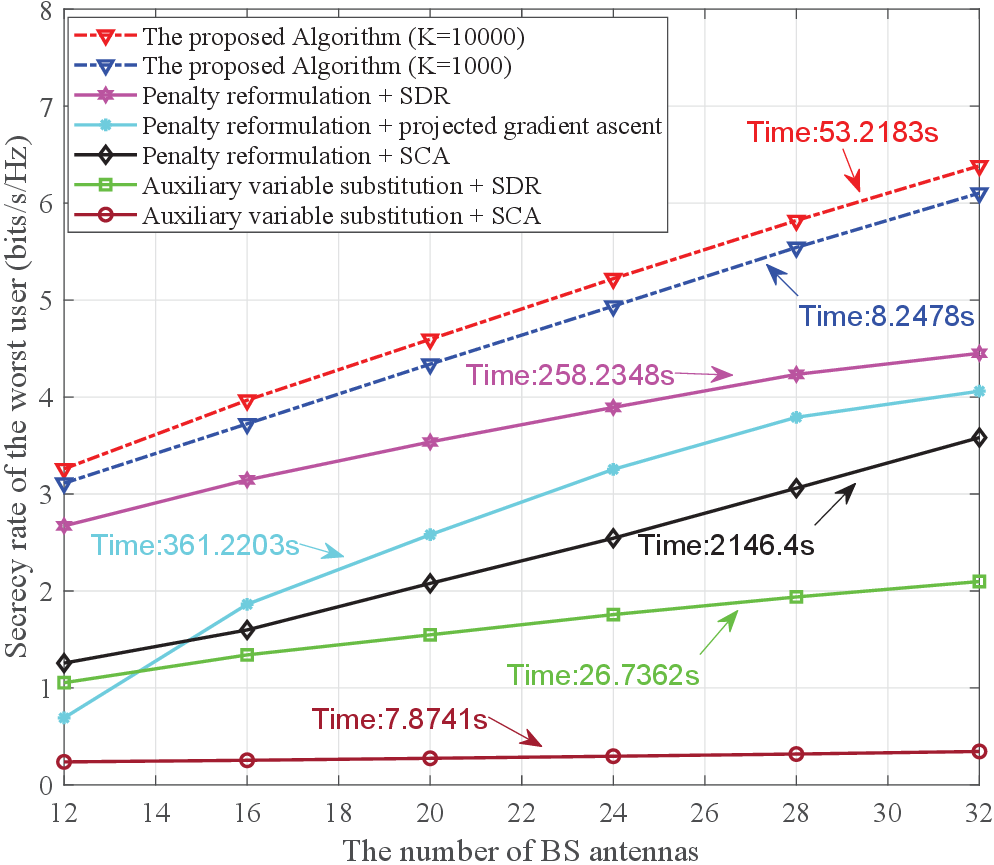}}\hfil
	\subfloat[]{\includegraphics[width=6.03cm]{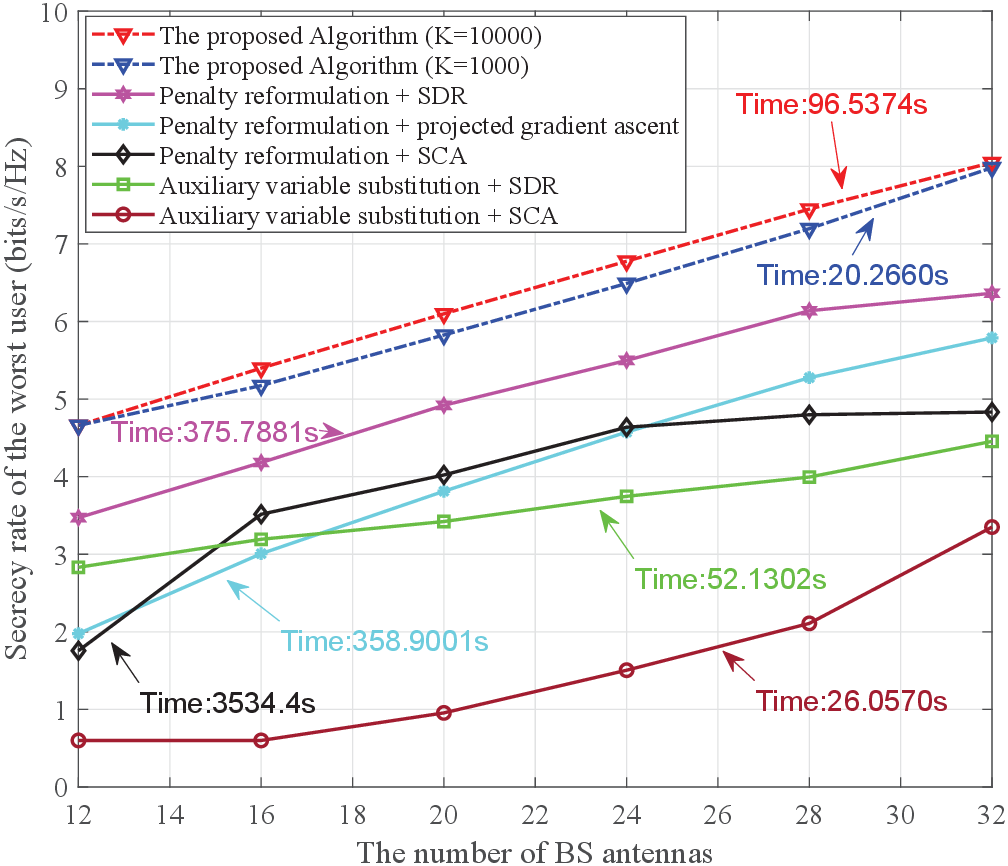}}\hfil 
	\subfloat[]{\includegraphics[width=6.03cm]{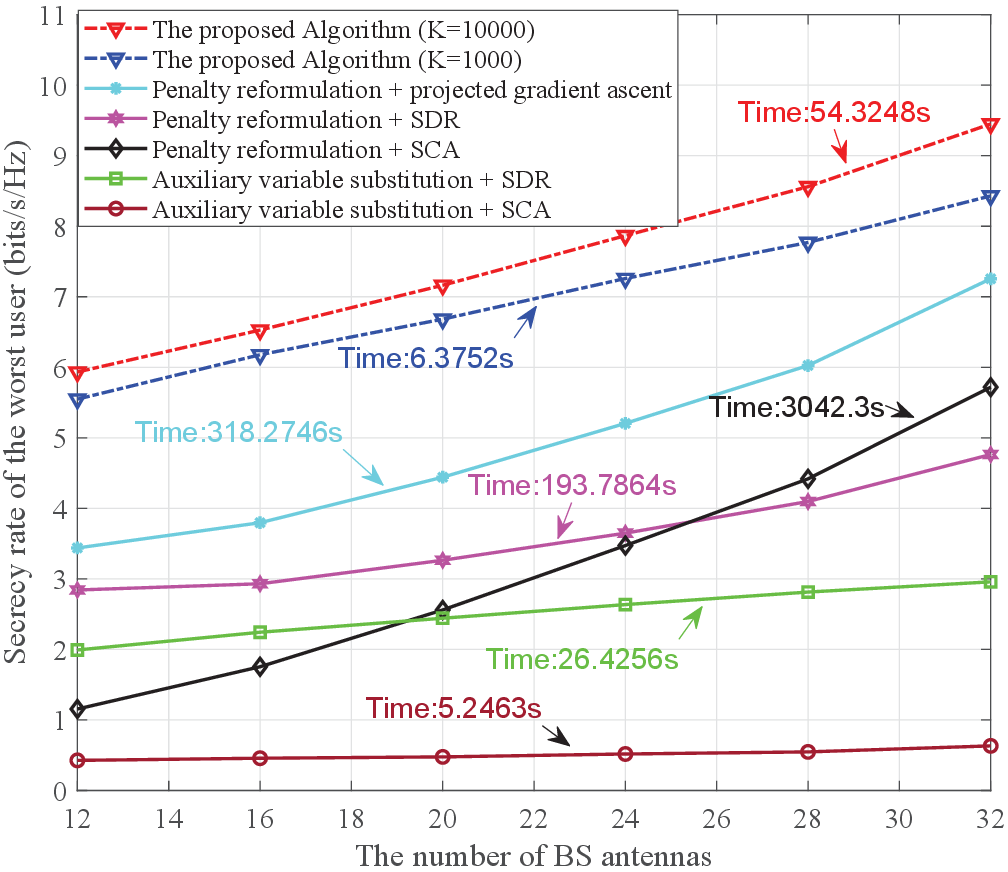}} 	
	\caption{Worst user secrecy rate versus number of BS antennas under three different sensing metrics. Simulation settings: (a) $I=7$, $J=9$, $M=6$, $P_{BS} = 20{\rm dBm}$ and $\Gamma  = 10{\rm dB}$ for sensing SINR constraint; (b)$I=5$, $J=7$, $M=4$, $P_{BS} = 20{\rm dBm}$ and $\gamma  = 0.1$ for beampattern matching constraint; (c) $I=6$, $N=28$, $M=7$, $P_{BS} = 15{\rm dBm}$, $P_{FA}=0.1$ and $\phi  = 0.85$ for detection probability constraint.}
	\label{fig_BS}
\end{figure*}

\begin{figure*}[t]
	\centering
	\subfloat[]{\includegraphics[width=6.03cm]{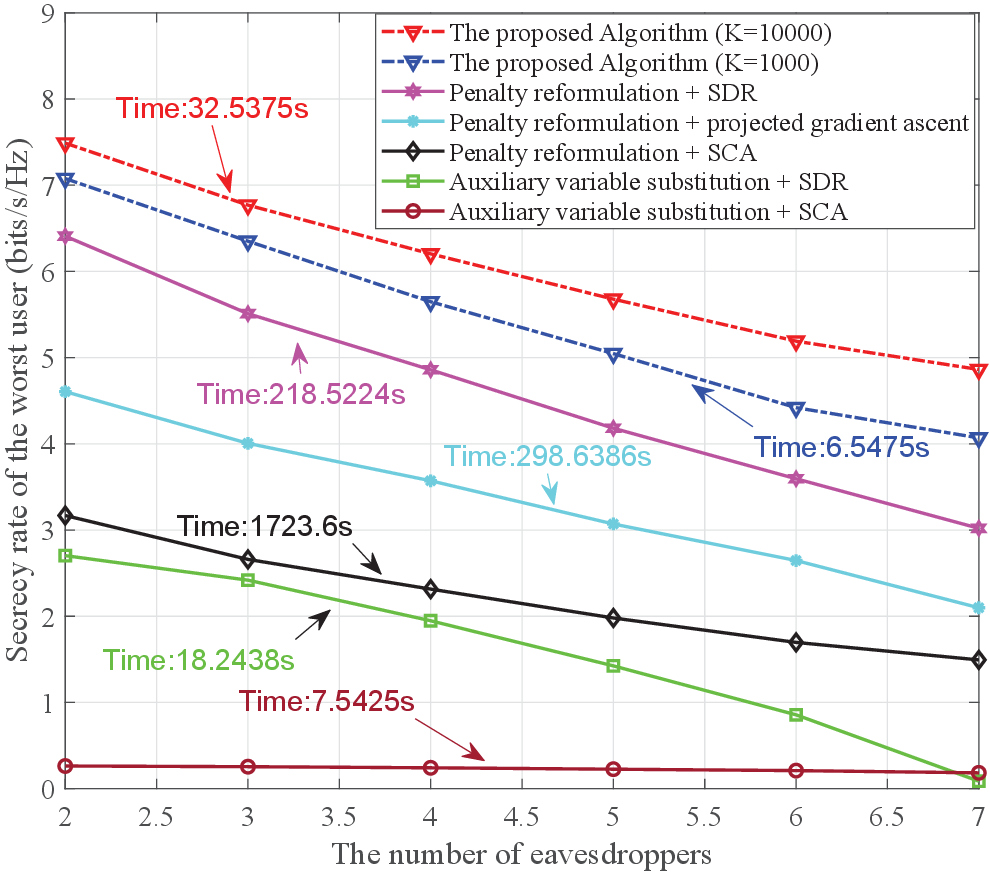}}\hfil
	\subfloat[]{\includegraphics[width=6.03cm]{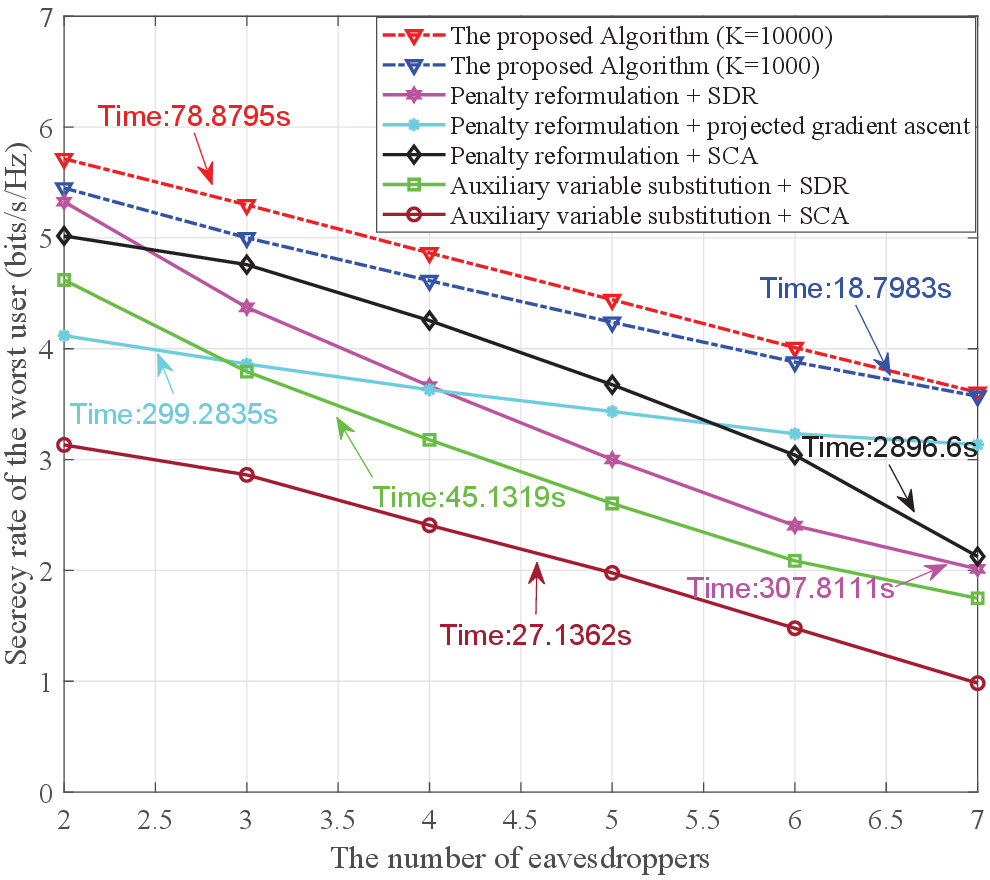}}\hfil 
	\subfloat[]{\includegraphics[width=6.03cm]{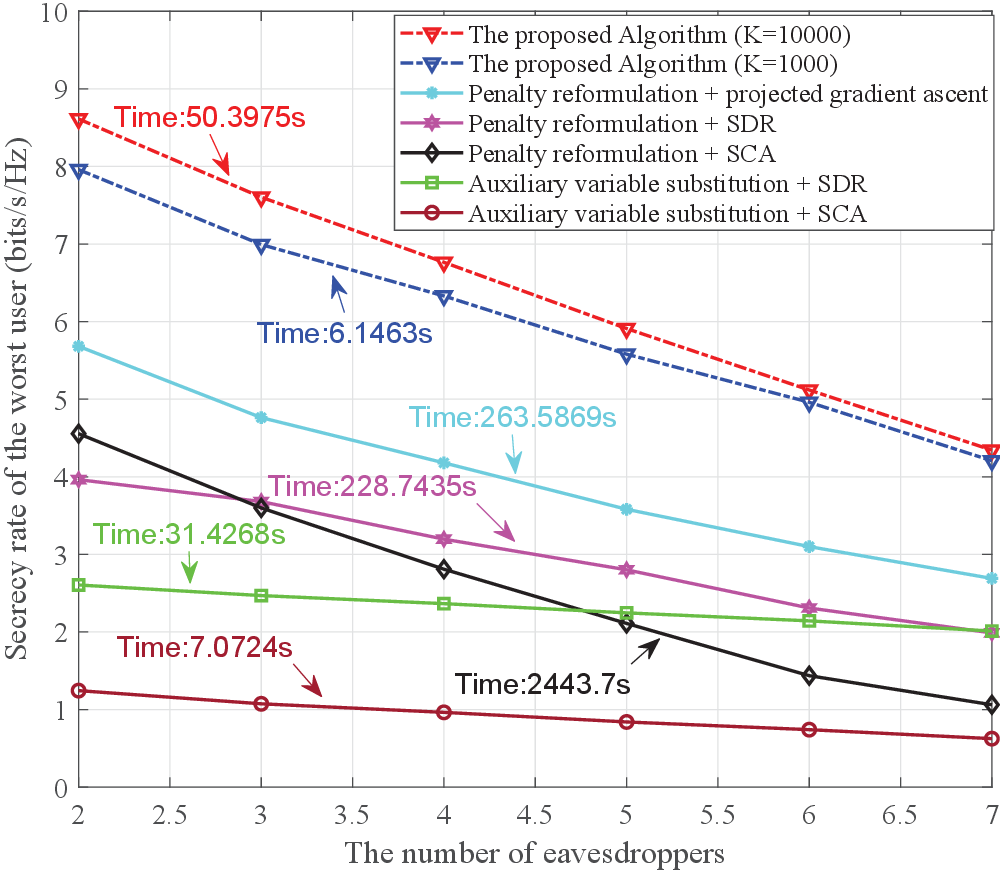}} 	
	\caption{Worst user secrecy rate versus number of eavesdroppers under three different sensing metrics. Simulation settings: (a) $I=7$, $N=24$, $M=6$, $P_{BS} = 15{\rm dBm}$ and $\Gamma  = 10{\rm dB}$ for sensing SINR constraint; (b) $I=5$, $N=24$, $M=4$, $P_{BS} = 10{\rm dBm}$ and $\gamma  = 0.1$ for beampattern matching constraint; (c) $I=6$, $J=4$, $M=7$, $P_{BS} = 20{\rm dBm}$, $P_{FA}=0.1$ and $\phi  = 0.9$ for detection probability constraint.}
	\label{fig_Eves}
\end{figure*}

\begin{figure*}[t]
	\centering
	\subfloat[]{\includegraphics[width=6.03cm]{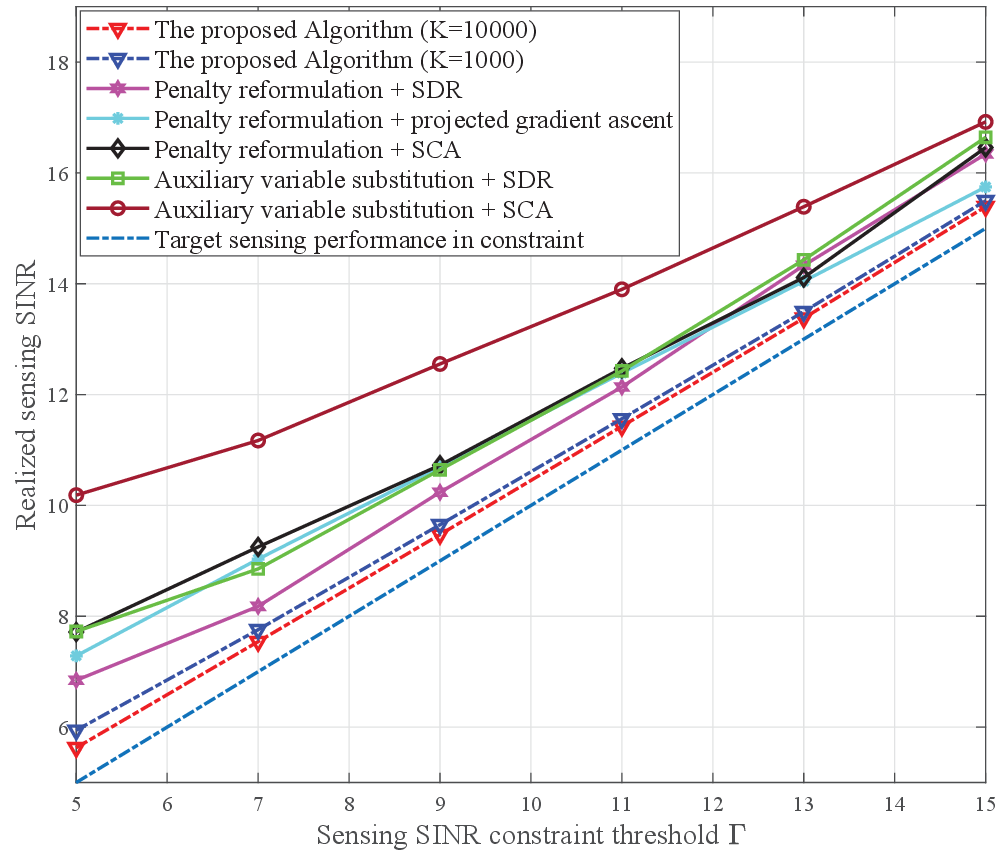}}\hfil
	\subfloat[]{\includegraphics[width=6.03cm]{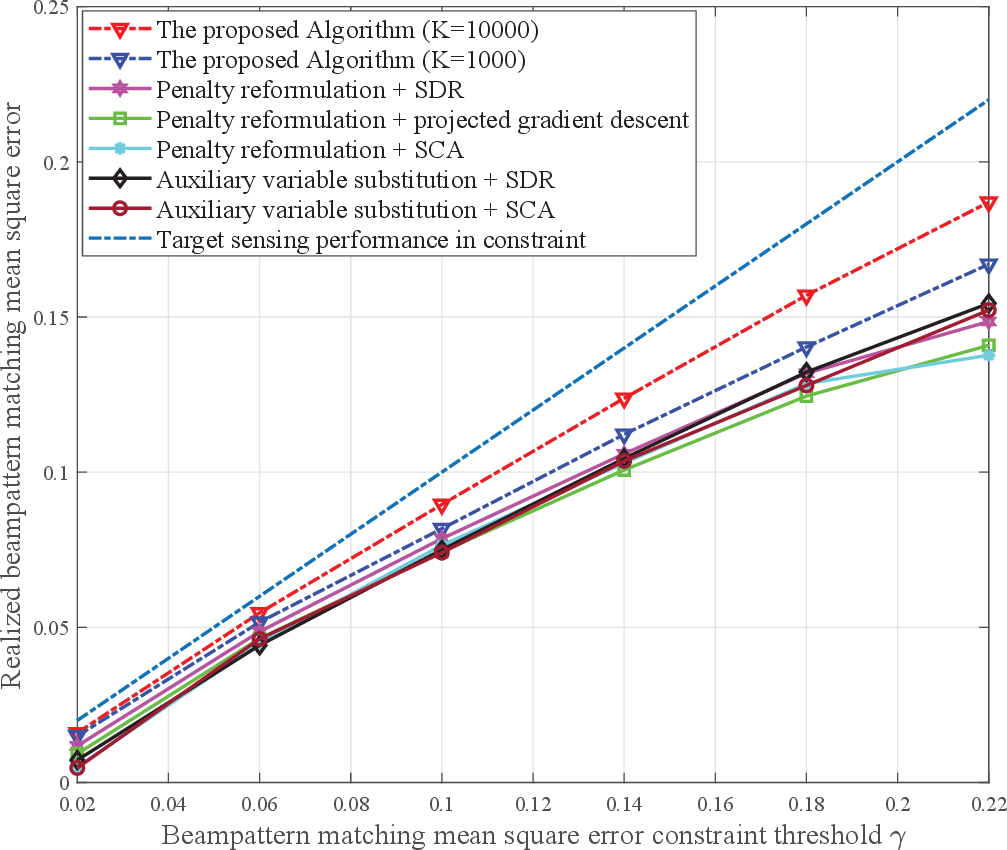}}\hfil 
	\subfloat[]{\includegraphics[width=6.03cm]{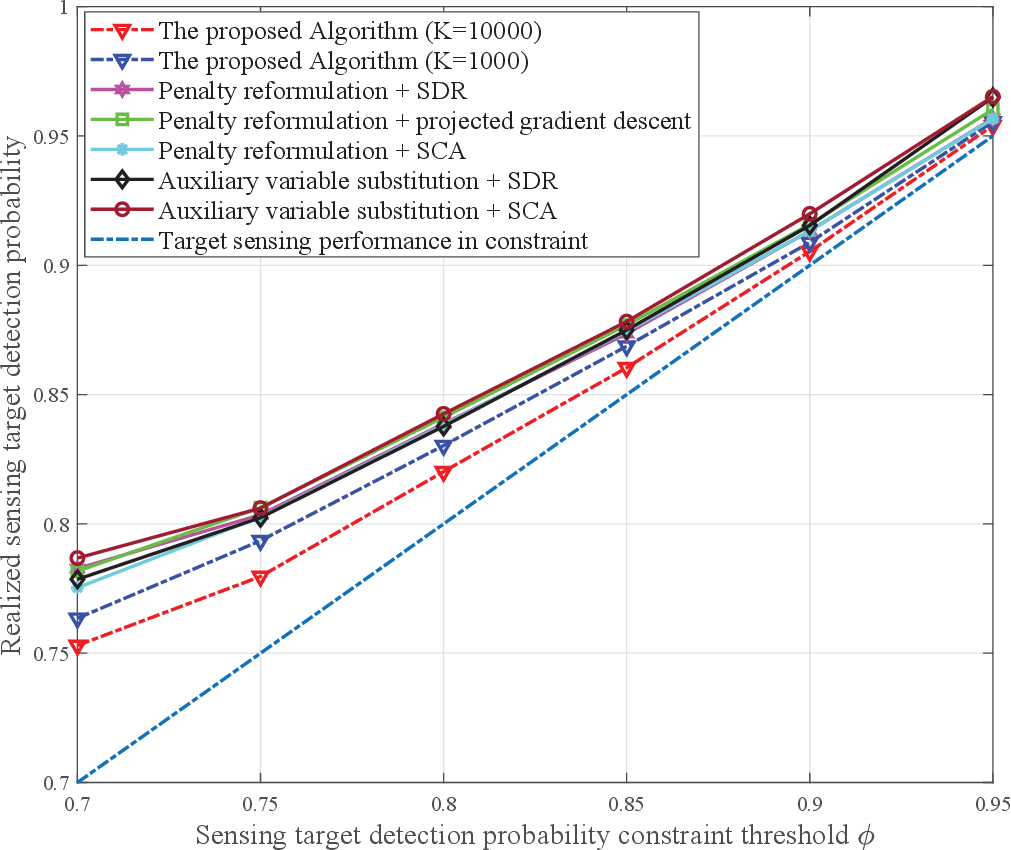}} 	
	\caption{Realized sensing performance under different sensing constraint thresholds (a) $N=20$, $I=6$, $J=7$, $M=5$, $P_{BS} = 20{\rm dBm}$ with sensing SINR metric, (b) $N=16$, $I=5$, $J=4$, $M=6$, $P_{BS} = 20{\rm dBm}$ with beampattern matching metric, (c) $N=12$, $I=6$, $J=5$, $M=4$, $P_{BS} = 20{\rm dBm}$, ${P_{FA}} = 0.1$ with detection probability metric}
	\label{figB1}
\end{figure*}

Besides the advantage of obtaining the best secrecy rate, the proposed algorithm also enjoys one of the lowest running times. Although the auxiliary variable substitution combined with SCA costs the least time, its secrecy rate is the worst and the solution is practically useless. In general, applying SCA and SDR rely on conventional convex tools such as CVX for solving the resulting problems, which lead to high computational complexity. For the penalty reformulation combined with projected gradient ascent method, while it uses the gradient information and thus mitigate the computation burden, it still requires explicit handling of the sensing performance constraint $S\left( {{\bm w},{\bm u}} \right) \ge 0$ in every iteration, which does not have closed-form solutions or simple one-dimensional bisection solutions as in the proposed algorithm. In contrast, the proposed algorithm only consists of computing gradients and simple calculations. As a result, the proposed algorithm costs significantly less time than the other methods.

Under the beampattern matching sensing criterion, the worst secrecy rates due to various methods are shown in Figs. \ref{fig_BS}(b) and \ref{fig_Eves}(b), which again shows the proposed algorithm is having the best performance. Notice that the beampattern matching sensing constraint is more complex than the SINR sensing constraint since it is a function containing the fourth power of beamformers $\left\{ {{{\bm w}_l}} \right\}_{l = 0}^I$ and its summation depends on the number of sampling angles $T$. This would increase the load of computing the gradients compared to the SINR constraint. In the compared methods, auxiliary variables are needed to replace the whole absolute value function \eqref{eq:20}, which would generate more auxiliary constraints and greatly slows down the compared methods. Although all the methods takes more time in beampattern matching constraint compared to the case of SINR constraint, calculating gradients is relatively more efficient than convexifying the complex sensing constraint. This leads to the proposed algorithm being one the fastest algorithms among the compared methods.

The worst user secrecy rate performance of various algorithms under detection probability sensing constraint are shown in Figs. \ref{fig_BS}(c) and \ref{fig_Eves}(c). Since the detection probability constraint is a summation of quadratic terms w.r.t. beamformers $\left\{ {{{\bm w}_l}} \right\}_{l = 0}^I$, the convex surrogate function in SCA can be chosen as the first-order Taylor expansion of these quadratic terms. Alternatively, we can transform these quadratic terms into matrix variables and then relaxe the problem to a semidefinite programming (SDP). Similar to the sensing SINR constraint and beampattern matching constraint cases, the proposed algorithm achieves the best secrecy rate performance with the one of the shortest execution times.

Furthermore, the realized sensing performance under different thresholds of the sensing constraints is shown in Fig. \ref{figB1}. From the figure, it can be seen that the proposed algorithm achieves realized sensing performance closest to the target specified in the sensing constraints (represented by the light blue line). Therefore, the proposed algorithm does not spend excessive resource to over-satisfy the sensing constraint, which leaves more resource for optimizing the secrecy rate. This explains the best secrecy rate achieved by the proposed algorithm compared to other existing methods in Figs. \ref{fig_BS} and \ref{fig_Eves}. The trade-off between sensing and secure communication performance is further shown in Fig. S1 in the supplementary material.

\begin{figure}[t]
	\centering
	\includegraphics[width=0.9\linewidth]{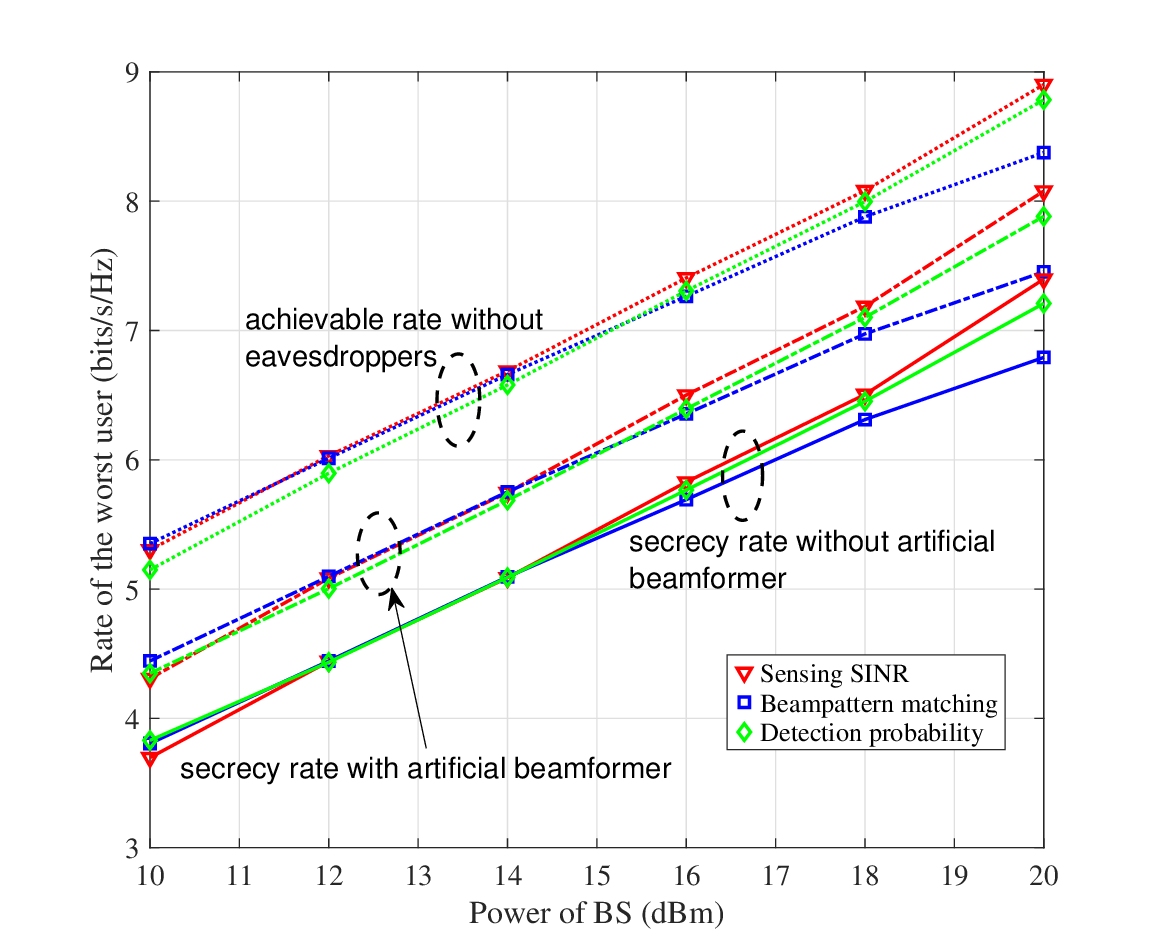}
	\caption{Worst user rate performance under three sensing metrics with $N=24$, $I=7$, $J=5$, $M=4$. $\Gamma  = 10{\rm dB}$ under sensing SINR constraint, $\gamma  = 0.1$ under beampattern matching constraint, ${P_{FA}} = 0.1$ and $\phi  = 0.9$ under detection probability constraint.}
	\label{fig_AN}
\end{figure}

Next, we show the impact of the eavesdroppers and the benefit of artificial beamformer $\bm w_0$ in Fig. \ref{fig_AN}. As expected, eavesdroppers would degrade the performance. Without artificial beamformer, the degradation could reach 33\% for all three considered sensing constraints. This shows that eavesdroppers cannot be neglected in ISAC. However, adopting artificial beamformer, which on the surface would take away some power resource from legitimate users, ends up improving the secrecy rate and leads to only 17\% degradation compared to no eavesdropper case. The artificial beamformer serves two purposes here: it provides a beam for sensing the target to facilitate the satisfaction of the sensing constraint and at the same time acts as artificial noise to counteract the eavesdroppers. 

\begin{figure*}[t]
	\centering
	\subfloat[]{\includegraphics[width=6.03cm]{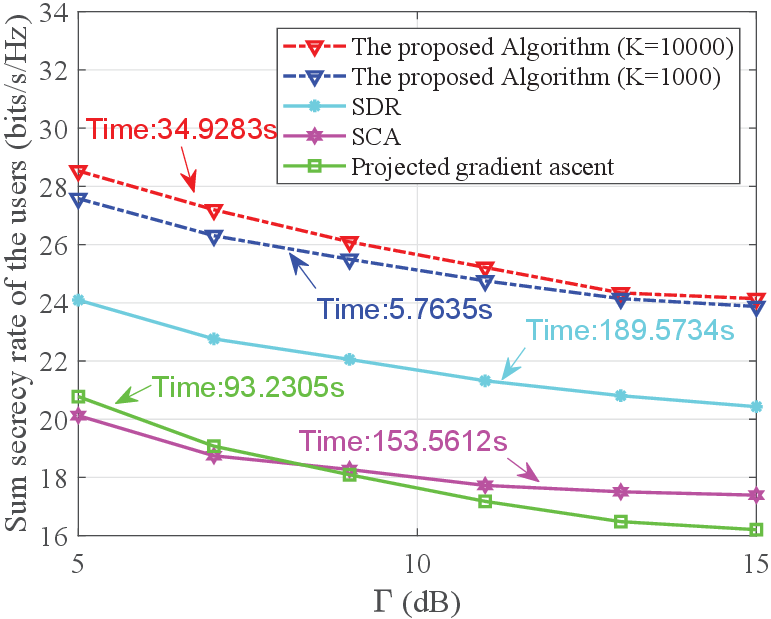}}\hfil
	\subfloat[]{\includegraphics[width=6.03cm]{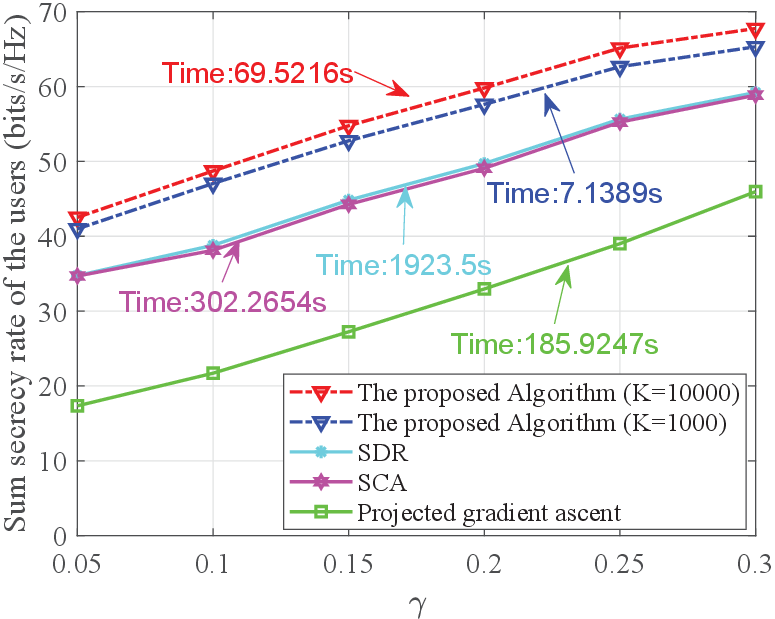}}\hfil 
	\subfloat[]{\includegraphics[width=6.03cm]{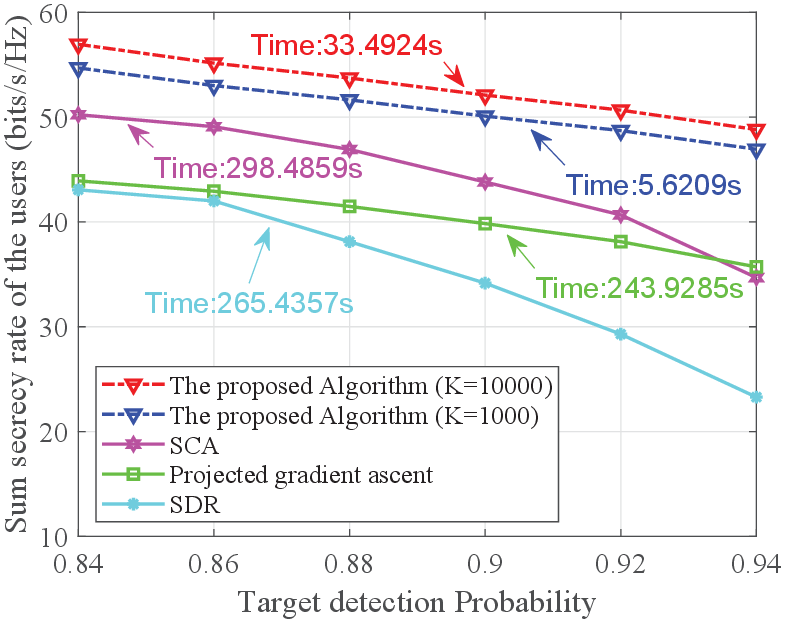}} 	
	\caption{Sum secrecy rate versus constraint thresholds under three different sensing metric. Simulation settings: (a) $N=16$, $I=4$, $J=5$, $M=3$, $P_{BS} = 10{\rm dBm}$ for sensing SINR constraint; (b) $N=24$, $I=7$, $J=5$, $M=6$, $P_{BS} = 15{\rm dBm}$ for beampattern matching constraint; (c) $N=24$, $I=6$, $J=7$, $M=5$, $P_{BS} = 15{\rm dBm}$ and $P_{FA}=0.1$ for detection probability constraint.}
	\label{figsumrate}
\end{figure*}

Finally, to show the versatility of the proposed general formulation and solution framework, Fig. \ref{figsumrate} illustrates the sum secrecy rate performance under three different sensing metrics with various sensing constraint thresholds. Due to the same advantages mentioned in worst user secrecy rate maximization, the proposed algorithm also enjoys better performance and lower running time compared to other algorithms in sum rate maximization.

\section{Conclusions}
This paper presented a unified optimization framework for maximizing the secrecy rate under various sensing constraints in secure ISAC systems. It was revealed that the probabilistic outage constraint can be handled precisely so that the limited resource could be better utilized for secrecy rate optimization. By transforming sensing constraints into the objective function and introducing a penalty variable, the proposed framework efficiently accommodates various sensing metrics, including (but not limited) to sensing SINR, beampattern matching, and detection probability. The resultant optimization problem was solved with AO algorithm enhanced with regularization, and  the solution was theoretically proved to converge to a stationary point. Simulation results validated the effectiveness and computational efficiency of the framework in optimizing the secrecy rates across various sensing criteria, outperforming other state-of-the-art methods. This unified framework not only streamlines algorithmic development for future research in secure ISAC systems but also facilitates straightforward comparisons across different scenarios.

\bibliographystyle{IEEEtran}

\bibliography{StarSystem}

@ARTICLE{Qunshu:25,
	author={Wang, Qunshu and Guo, Shaoyong and Wu, Celimuge and Xing, Chengwen and Zhao, Nan and Niyato, Dusit and Karagiannidis, George K.},
	journal={IEEE J. Sel. Areas Commun.}, 
	title={STAR-RIS Aided Covert Communication in UAV Air-Ground Networks}, 
	year={2025},
	month={Sep},
	volume={43},
	number={1},
	pages={245-259},
	doi={10.1109/JSAC.2024.3460079}}

@ARTICLE{Shengbin:24,
	author={Lin, Shengbin and Xu, Yitao and Wang, Haichao and Ding, Guoru},
	journal={IEEE Trans. Wireless Commun.}, 
	title={Multi-Antenna Covert Communication Assisted by UAV-RIS With Imperfect CSI}, 
	year={2024},
	month={Jun},
	volume={23},
	number={10},
	pages={13841-13855},
	doi={10.1109/TWC.2024.3405201}}

@Article{Xuanxuan:19,
	author={Wang, Xuanxuan and Feng, Wei and Chen, Yunfei and Ge, Ning},
	journal={Proc. - IEEE GLOBECOM}, 
	title={Power Allocation for {UAV} Swarm-Enabled Secure Networks Using Large-Scale CSI}, 
	year={2019},
	month={Dec},
	volume={},
	number={},
	pages={1-6},
	doi={10.1109/GLOBECOM38437.2019.9014165}}

@article{Jingke:25,
	title={Dual-UAV-Aided Covert Communications for Air-to-Ground ISAC Networks},
	author={Sun, Jingke and Yang, Liang and Boulogeorgos, Alexandros-Apostolos A and Tsiftsis, Theodoros A and Liu, Hongwu},
	journal={arXiv preprint arXiv:2506.00601},
	year={2025}
}

@article{Baogang:24,
	title={{UAV}-assisted {ISAC} network physical layer security based on Stackelberg game},
	author={Li, Baogang and Liao, Jia and Gong, Xi and Xiang, Hongyin and Yang, Zhi and Zhao, Wei},
	journal={Ad Hoc Netw.},
	volume={152},
	pages={103304},
	year={2024},
	month={Jan},
	publisher={Elsevier}
}

@Article{Yulong:13,
	author={Zou, Yulong and Wang, Xianbin and Shen, Weiming},
	journal={Proc. - IEEE ICC}, 
	title={Intercept probability analysis of cooperative wireless networks with best relay selection in the presence of eavesdropping attack}, 
	year={2013},
	month={Jul},
	volume={},
	number={},
	pages={2183-2187},
	doi={10.1109/ICC.2013.6654851}}

@Article{Jiaxin:15,
	author={Yang, Jiaxin and Champagne, Benoit and Li, Qiang and Hanzo, Lajos},
	journal={Proc. - IEEE GLOBECOM}, 
	title={Secure {MIMO} {AF} Relaying Design: An Intercept Probability Constrained Approach}, 
	year={2015},
	month={Dec},
	volume={},
	number={},
	pages={1-6},
	doi={10.1109/GLOCOM.2015.7417241}}

@ARTICLE{Giovanni:25,
	author={Interdonato, Giovanni and Di Murro, Francesca and D’Andrea, Carmen and Di Gennaro, Giovanni and Buzzi, Stefano},
	journal={IEEE Trans. Commun.}, 
	title={Approaching Massive MIMO Performance With Reconfigurable Intelligent Surfaces: We Do Not Need Many Antennas}, 
	year={2025},
	month={Nov},
	volume={73},
	number={6},
	pages={4000-4016},
	doi={10.1109/TCOMM.2024.3506928}}

@article{xu2025mutual,
	title={Mutual Information-Oriented ISAC Beamforming Design for Large Dimensional Antenna Array},
	author={Xu, Shanfeng and Cheng, Yanshuo and Wang, Siqiang and Wang, Xinyi and Zheng, Zhong and Fei, Zesong},
	journal={Electronics},
	volume={14},
	number={13},
	pages={2515},
	year={2025},
	month={Jun},
	publisher={MDPI}
}

@ARTICLE{Yiqing:22,
	author={Li, Yiqing and Jiang, Miao},
	journal={IEEE Trans. Veh. Technol.}, 
	title={Joint Transmit Beamforming and Receive Filters Design for Coordinated Two-Cell Interfering Dual-Functional Radar-Communication Networks}, 
	year={2022},
	month={Jul},
	volume={71},
	number={11},
	pages={12362-12367},
	doi={10.1109/TVT.2022.3193001}}

@ARTICLE{Na:22,
	author={Zhao, Na and Wang, Yunlong and Zhang, Zhibo and Chang, Qing and Shen, Yuan},
	journal={IEEE Commun. Lett.}, 
	title={Joint Transmit and Receive Beamforming Design for Integrated Sensing and Communication}, 
	year={2022},
	month={Mar},
	volume={26},
	number={3},
	pages={662-666},
	doi={10.1109/LCOMM.2021.3140093}}

@article{pritzker2022transmit,
	title={Transmit precoder design approaches for dual-function radar-communication systems},
	author={Pritzker, Jacob and Ward, James and Eldar, Yonina C},
	journal={arXiv preprint arXiv:2203.09571},
	year={2022}
}

@ARTICLE{Yuanhan:22,
	author={Ni, Yuanhan and Wang, Zulin and Huang, Qin},
	journal={IEEE Wireless Commun. Lett.}, 
	title={Joint Transceiver Beamforming for Multi-Target Single-User Joint Radar and Communication}, 
	year={2022},
	month={Nov},
	volume={11},
	number={11},
	pages={2360-2364},
	doi={10.1109/LWC.2022.3203386}}

@ARTICLE{Chen:22a,
	author={Chen, Li and Wang, Zhiqin and Du, Ying and Chen, Yunfei and Yu, F. Richard},
	journal={IEEE J. Sel. Areas Commun.}, 
	title={Generalized Transceiver Beamforming for DFRC With MIMO Radar and MU-MIMO Communication}, 
	year={2022},
	month={Jun},
	volume={40},
	number={6},
	pages={1795-1808},
	doi={10.1109/JSAC.2022.3155515}}

@ARTICLE{Xiang:22b,
	author={Liu, Xiang and Huang, Tianyao and Liu, Yimin},
	journal={IEEE J. Sel. Areas Commun.}, 
	title={Transmit Design for Joint MIMO Radar and Multiuser Communications With Transmit Covariance Constraint}, 
	year={2022},
	month={Mar},
	volume={40},
	number={6},
	pages={1932-1950},
	doi={10.1109/JSAC.2022.3155512}}

@ARTICLE{Xiang:20,
	author={Liu, Xiang and Huang, Tianyao and Shlezinger, Nir and Liu, Yimin and Zhou, Jie and Eldar, Yonina C.},
	journal={IEEE Trans. Signal Process.}, 
	title={Joint Transmit Beamforming for Multiuser {MIMO} Communications and {MIMO} Radar}, 
	year={2020},
	month={Jun},
	volume={68},
	number={},
	pages={3929-3944},
	doi={10.1109/TSP.2020.3004739}}

@ARTICLE{Fan:18,
	author={Liu, Fan and Masouros, Christos and Li, Ang and Sun, Huafei and Hanzo, Lajos},
	journal={IEEE Trans. Wireless Commun.}, 
	title={MU-MIMO Communications With MIMO Radar: From Co-Existence to Joint Transmission}, 
	year={2018},
	month={Feb},
	volume={17},
	number={4},
	pages={2755-2770},
	doi={10.1109/TWC.2018.2803045}}

@ARTICLE{Zechen:24,
	author={Liu, Zechen and Liu, Xin and Liu, Yuemin and Leung, Victor C. M. and Durrani, Tariq S.},
	journal={IEEE Trans. Wireless Commun.}, 
	title={UAV Assisted Integrated Sensing and Communications for Internet of Things: 3D Trajectory Optimization and Resource Allocation}, 
	year={2024},
	month={Jan},
	volume={23},
	number={8},
	pages={8654-8667},
	doi={10.1109/TWC.2024.3352985}}

@Article{Xipeng:24,
	author={Chen, Xipeng and Cao, Xiaowen and Xie, Lifeng and He, Yejun},
	journal={Proc - IEEE iWRFAT}, 
	title={{DRL}-Based Joint Trajectory Planning and Beamforming Optimization in Aerial RIS-Assisted ISAC System}, 
	year={2024},
	month={Jul},
	volume={},
	number={},
	pages={510-515},
	doi={10.1109/iWRFAT61200.2024.10594249}}

@ARTICLE{Salem:2025,
	author={Salem, A. Abdelaziz and Abdallah, Saeed and Saad, Mohamed and Alnajjar, Khawla and Albreem, Mahmoud A.},
	journal={IEEE Trans. Veh. Technol.}, 
	title={Robust Secure ISAC: How RSMA and Active RIS Manage Eavesdropper's Spatial Uncertainty}, 
	year={2025},
	month={Sep},
	volume={},
	number={},
	pages={1-16},
	doi={10.1109/TVT.2025.3610863}}

@ARTICLE{Haitao:23,
	author={Zhao, Haitao and Wu, Fengjing and Xia, Wenchao and Zhang, Yao and Ni, Yiyang and Zhu, Hongbo},
	journal={IEEE Commun. Lett.}, 
	title={Joint Beamforming Design for RIS-Aided Secure Integrated Sensing and Communication Systems}, 
	year={2023},
	month={Sep},
	volume={27},
	number={11},
	pages={2943-2947},
	doi={10.1109/LCOMM.2023.3312089}}

@ARTICLE{Hui-Ming:15,
	author={Wang, Hui-Ming and Zheng, Tongxing and Xia, Xiang-Gen},
	journal={IEEE Trans. Wireless Commun.}, 
	title={Secure MISO Wiretap Channels With Multiantenna Passive Eavesdropper: Artificial Noise vs. Artificial Fast Fading}, 
	year={2015},
	month={Jun},
	volume={14},
	number={1},
	pages={94-106},
	doi={10.1109/TWC.2014.2332164}}

@ARTICLE{Qi:14,
	author={Xiong, Qi and Gong, Yi and Liang, Ying-Chang and Li, Kwok Hung},
	journal={IEEE Wireless Commun. Lett.}, 
	title={Achieving Secrecy of MISO Fading Wiretap Channels via Jamming and Precoding With Imperfect Channel State Information}, 
	year={2014},
	month={Apr},
	volume={3},
	number={4},
	pages={357-360},
	doi={10.1109/LWC.2014.2317194}}

@ARTICLE{Xi:15,
	author={Zhang, Xi and McKay, Matthew R. and Zhou, Xiangyun and Heath, Robert W.},
	journal={IEEE Trans. Wireless Commun.}, 
	title={Artificial-Noise-Aided Secure Multi-Antenna Transmission With Limited Feedback}, 
	year={2015},
	month={Jan},
	volume={14},
	number={5},
	pages={2742-2754},
	doi={10.1109/TWC.2015.2391261}}

@ARTICLE{Wei:17,
	author={Wang, Wei and Teh, Kah Chan and Li, Kwok Hung},
	journal={IEEE Trans. Inf. Forensics Secur.}, 
	title={Secrecy Throughput Maximization for MISO Multi-Eavesdropper Wiretap Channels}, 
	year={2017},
	month={Oct},
	volume={12},
	number={3},
	pages={505-515},
	doi={10.1109/TIFS.2016.2620279}}

@ARTICLE{Meng:24,
	author={Hua, Meng and Wu, Qingqing and Chen, Wen and Dobre, Octavia A. and Swindlehurst, A. Lee},
	journal={IEEE Trans. Wireless Commun.}, 
	title={Secure Intelligent Reflecting Surface-Aided Integrated Sensing and Communication}, 
	year={2024},
	month={Jun},
	volume={23},
	number={1},
	pages={575-591},
	doi={10.1109/TWC.2023.3280179}}

@ARTICLE{Chen:25,
	author={Zhang, Chen and Qu, Shaocheng and Zhao, Li and Wei, Ziming and Shi, Qianqian and Liu, Youya},
	journal={IEEE Wireless Commun. Lett.}, 
	title={Robust Secure Beamforming Design for Downlink RIS-ISAC Systems Enhanced by RSMA}, 
	year={2025},
	month={Dec},
	volume={14},
	number={5},
	pages={1271-1275},
	doi={10.1109/LWC.2024.3521211}}

@ARTICLE{Yue:25,
	author={Xiu, Yue and Lyu, Wanting and Yeoh, Phee Lep and Ai, Yi and Wei, Ning},
	journal={IEEE Trans. Veh. Technol.}, 
	title={Secure Enhancement for RIS-Aided UAV with ISAC: Robust Design and Resource Allocation}, 
	year={2025},
	month={Oct},
	volume={},
	number={},
	pages={1-16},
	doi={10.1109/TVT.2025.3623453}}

@ARTICLE{Ziwei:25,
	author={Wang, Ziwei and Lei, Xianfu and Mathiopoulos, P. Takis and Wang, Gongpu},
	journal={IEEE Trans. Veh. Technol.}, 
	title={Transmit Power Optimization of Cooperative NOMA-Aided Secure ISAC Systems With SWIPT}, 
	year={2025},
	month={Mar},
	volume={74},
	number={8},
	pages={13207-13212},
	doi={10.1109/TVT.2025.3553908}}

@ARTICLE{SuNanchi:21,
	author={Su, Nanchi and Liu, Fan and Masouros, Christos},
	journal={IEEE Trans. Wireless Commun.}, 
	title={Secure Radar-Communication Systems With Malicious Targets: Integrating Radar, Communications and Jamming Functionalities}, 
	year={2021},
	month={Sep},
	volume={20},
	number={1},
	pages={83-95},
	doi={10.1109/TWC.2020.3023164}}

@ARTICLE{Yuchen:25,
	author={Zhang, Yuchen and Ni, Wanli and Wang, Jianquan and Tang, Wanbin and Jia, Min and Eldar, Yonina C. and Niyato, Dusit},
	journal={IEEE Transactions on Communications}, 
	title={Robust Transceiver Design for Covert Integrated Sensing and Communications With Imperfect CSI}, 
	year={2025},
	month={Apr},
	volume={73},
	number={9},
	pages={8016-8031},
	doi={10.1109/TCOMM.2024.3387869}}

@ARTICLE{Hanbo:24,
	author={Jia, Hanbo and Ma, Lin and Qin, Danyang},
	journal={IEEE Trans. Veh. Technol.}, 
	title={Robust Beamforming Design for Covert Integrated Sensing and Communication in the Presence of Multiple Wardens}, 
	year={2024},
	month={Jul},
	volume={73},
	number={11},
	pages={17135-17150},
	doi={10.1109/TVT.2024.3425960}}

@ARTICLE{Xingyu:24,
	author={Zhao, Xingyu and Deng, Weicao and Li, Min and Zhao, Min-Jian},
	journal={IEEE Wireless Commun. Lett.}, 
	title={Robust Beamforming Design for Integrated Sensing and Covert Communication Systems}, 
	year={2024},
	month={Jul},
	volume={13},
	number={9},
	pages={2566-2570},
	doi={10.1109/LWC.2024.3429149}}

@ARTICLE{Zhutian:22,
	author={Yang, Zhutian and Li, Dongdong and Zhao, Nan and Wu, Zhilu and Li, Yonghui and Niyato, Dusit},
	journal={IEEE Trans. Commun.}, 
	title={Secure Precoding Optimization for NOMA-Aided Integrated Sensing and Communication}, 
	year={2022},
	month={Oct},
	volume={70},
	number={12},
	pages={8370-8382},
	doi={10.1109/TCOMM.2022.3216636}}

@ARTICLE{Hou:24,
	author={Hou, Kaiyue and Zhang, Shuowen},
	journal={IEEE J. Sel. Areas Commun.}, 
	title={Optimal Beamforming for Secure Integrated Sensing and Communication Exploiting Target Location Distribution}, 
	year={2024},
	month={Aug},
	volume={42},
	number={11},
	pages={3125-3139},
	doi={10.1109/JSAC.2024.3431573}}

@ARTICLE{Ziang:24,
	author={Liu, Ziang and Yin, Longfei and Shin, Wonjae and Clerckx, Bruno},
	journal={IEEE Trans. Wireless Commun.}, 
	title={Rate-Splitting Multiple Access for Quantized ISAC LEO Satellite Systems: A Max-Min Fair Energy-Efficient Beam Design}, 
	year={2024},
	month={Jul},
	volume={23},
	number={10},
	pages={15394-15408},
	doi={10.1109/TWC.2024.3429229}}

@Article{Kun-Yu:11,
	author={Wang, Kun-Yu and Chang, Tsung-Hui and Ma, Wing-Kin and So, Anthony Man-Cho and Chi, Chong-Yung},
	journal={Proc. - ICASSP 2011}, 
	title={Probabilistic {SINR} constrained robust transmit beamforming: A Bernstein-type inequality based conservative approach}, 
	year={2011},
	month={May},
	volume={},
	number={},
	pages={3080-3083},
	doi={10.1109/ICASSP.2011.5946309}}

@ARTICLE{Xuehua:25,
	author={Li, Xuehua and Wu, Zhongqing and Cai, Yuanxin and Hu, Shaokang and Liao, Yihuan and Yuan, Weijie},
	journal={IEEE J. Sel. Areas Commun.}, 
	title={Win-win of Communication and Sensing Security for MC-NOMA ISAC Systems}, 
	year={2025},
	month={May},
	volume={43},
	number={9},
	pages={3214 - 3230},
	doi={10.1109/JSAC.2025.3574620}}

@Article{Wenyi:24,
	author={Yang, Wenyi and Liu, Xin and Liu, Zechen},
	journal={Proc. - ICCT 2024}, 
	title={Joint Trajectory and Resource Allocation Optimization for Mobile Vehicles in UAV Assisted Secure ISAC System}, 
	year={2024},
	month={Oct},
	volume={},
	number={},
	pages={321-325},
	doi={10.1109/ICCT62411.2024.10946325}}

@ARTICLE{Yuemin:24,
	author={Liu, Yuemin and Liu, Xin and Liu, Zechen and Yu, Yingfeng and Jia, Min and Na, Zhenyu and Durrani, Tariq S},
	journal={IEEE Trans. Veh. Technol.}, 
	title={Secure Rate Maximization for ISAC-UAV Assisted Communication Amidst Multiple Eavesdroppers}, 
	year={2024},
	month={Oct},
	volume={73},
	number={10},
	pages={15843-15847},
	doi={10.1109/TVT.2024.3412805}}

@Article{Tejaswini:24,
	author={Lakkimsetti, Tejaswini and Rao Ukyam, Uma Maheswara and Kumar Gurrala, Kiran},
	journal={Proc. - INDICON 2024}, 
	title={Power Allocation for Secure {RIS}-Aided {ISAC} {NOMA} Network}, 
	year={2024},
	month={Dec},
	volume={},
	number={},
	pages={1-6},
	doi={10.1109/INDICON63790.2024.10958515}}

@Article{Ruiwei:23,
	author={Yang, Ruiwei and Du, Huiqin},
	journal={Proc. - SPAWC 2023}, 
	title={Joint Precoding and Artificial Noise Design for Secure Transmission in ISAC system}, 
	year={2023},
	month={Nov},
	volume={},
	number={},
	pages={16-20},
	doi={10.1109/SPAWC53906.2023.10304462}}

@Article{Wai-Yiu:23,
	author={Keung, Wai-Yiu and Wai, Hoi-To and Ma, Wing-Kin},
	journal={Proc. - ICASSPW 2023}, 
	title={Secure Integrated Sensing and Communication Downlink Beamforming: A Semidefinite Relaxation Approach With Tightness Guaranteed}, 
	year={2023},
	month={Aug},
	volume={},
	number={},
	pages={1-5},
	doi={10.1109/ICASSPW59220.2023.10193088}}

@Article{Nanchi:23,
	author={Su, Nanchi and Liu, Fan and Masouros, Christos},
	journal={Proc. - EUSIPCO 2023}, 
	title={Secure Integrated Sensing and Communication Systems with the Assistance of Sensing Functionality}, 
	year={2023},
	month={Nov},
	volume={},
	number={},
	pages={690-694},
	doi={10.23919/EUSIPCO58844.2023.10289830}}

@ARTICLE{Silei:25,
	author={He, Silei and Tang, Kun and Zheng, Beixiong and Xiu, Xin and Feng, Wenjie and Che, Wenquan and Xue, Quan},
	journal={IEEE Trans. Commun.}, 
	title={Throughput Maximization Design for {RIS}-Assisted {WPCN}-{NOMA}-Based {ISAC} Systems}, 
	year={2025},
	month={Feb},
	volume={73},
	number={8},
	pages={6943 - 6957},
	doi={10.1109/TCOMM.2025.3543242}}

@book{Tse:05,
	title={Fundamentals of Wireless Communication},
	author={Tse, David and Viswanath, Pramod},
	year={2005},
	publisher={Cambridge University Press}
}

@ARTICLE{Zhichao:18,
	author={Sheng, Zhichao and Tuan, Hoang Duong and Duong, Trung Q. and Poor, H. Vincent},
	journal={IEEE Trans. Signal Process.}, 
	title={Beamforming Optimization for Physical Layer Security in MISO Wireless Networks}, 
	year={2018},
	month={May},
	volume={66},
	number={14},
	pages={3710-3723},
	doi={10.1109/TSP.2018.2835406}}

@article{Zhijun:24,
	title={Joint jammer selection and power optimization in covert communications against a warden with uncertain locations},
	author={Han, Zhijun and Zhou, Yiqing and Zhang, Yu and Zheng, Tong-Xing and Liu, Ling and Shi, Jinglin},
	journal={Digit. Commun. Netw.},
	year={2025},
	month={Aug},
	volume={11},
	number={4},
	pages={1114-1124},
	publisher={Elsevier}
}

@ARTICLE{Shihang:24,
	author={Lu, Shihang and Liu, Fan and Li, Yunxin and Zhang, Kecheng and Huang, Hongjia and Zou, Jiaqi and Li, Xinyu and Dong, Yuxiang and Dong, Fuwang and Zhu, Jia and Xiong, Yifeng and Yuan, Weijie and Cui, Yuanhao and Hanzo, Lajos},
	journal={IEEE Internet Things J.}, 
	title={Integrated Sensing and Communications: Recent Advances and Ten Open Challenges}, 
	year={2024},
	month={Feb},
	volume={11},
	number={11},
	pages={19094-19120},
	doi={10.1109/JIOT.2024.3361173}}

@Article{Bliss:14,
	author={Bliss, Daniel W.},
	journal={Proc.- IEEE Natl. Radar Conf.}, 
	title={Cooperative radar and communications signaling: The estimation and information theory odd couple}, 
	volume={},
	number={},
	pages={50-55},
	year={2014},
	month={May},
	doi={10.1109/RADAR.2014.6875553}}

@ARTICLE{Tan:18,
	author={Tan, Bo and Chen, Qingchao and Chetty, Kevin and Woodbridge, Karl and Li, Wenda and Piechocki, Robert},
	journal={IEEE Commun. Mag.}, 
	title={Exploiting WiFi Channel State Information for Residential Healthcare Informatics}, 
	year={2018},
	month={May},
	volume={56},
	number={5},
	pages={130-137},
	doi={10.1109/MCOM.2018.1700064}}

@ARTICLE{Zongze:21,
	author={Li, Zongze and Xia, Minghua and Wen, Miaowen and Wu, Yik-Chung},
	journal={IEEE J. Sel. Areas Commun.}, 
	title={Massive Access in Secure NOMA Under Imperfect CSI: Security Guaranteed Sum-Rate Maximization With First-Order Algorithm}, 
	year={2021},
	month={Aug},
	volume={39},
	number={4},
	pages={998-1014},
	doi={10.1109/JSAC.2020.3018805}}

@Article{Kaiqing:21,
	title={Decentralized policy gradient descent ascent for safe multi-agent reinforcement learning},
	author={Lu, Songtao and Zhang, Kaiqing and Chen, Tianyi and Ba{\c{s}}ar, Tamer and Horesh, Lior},
	journal={Proc. - AAAI Conf. Artif. Intell.},
	volume={35},
	number={10},
	pages={8767--8775},
	year={2021}
}

@article{Bin:06,
	title={Gradient methods with adaptive step-sizes},
	author={Zhou, Bin and Gao, Li and Dai, Yu-Hong},
	journal={Comput. Optim. Appl.},
	volume={35},
	pages={69--86},
	year={2006},
	month={Mar},
	publisher={Springer}
}

@book{beck:17,
	author = {Beck, Amir},
	title = {First-Order Methods in Optimization},
	publisher = {Society for Industrial and Applied Mathematics},
	year = {2017},
	doi = {10.1137/1.9781611974997},
	address = {Philadelphia, PA}
}

@Article{Tianyi:20,
	title={Near-optimal algorithms for minimax optimization},
	author={Lin, Tianyi and Jin, Chi and Jordan, Michael I},
	journal={Proc. - Conf. Learn. Theory.},
	pages={2738--2779},
	year={2020},
	organization={PMLR}
}

@article{Zi:23,
	title={A unified single-loop alternating gradient projection algorithm for nonconvex--concave and convex--nonconcave minimax problems},
	author={Xu, Zi and Zhang, Huiling and Xu, Yang and Lan, Guanghui},
	journal={Math. Program.},
	volume={201},
	number={1},
	pages={635--706},
	year={2023},
	month={Jan},
	publisher={Springer}
}

@article{Jiawei:20,
	title={A single-loop smoothed gradient descent-ascent algorithm for nonconvex-concave min-max problems},
	author={Zhang, Jiawei and Xiao, Peijun and Sun, Ruoyu and Luo, Zhiquan},
	journal={Proc. - Adv. Neural Inf. Process. Syst. (NeurIPS)},
	volume={33},
	pages={7377--7389},
	year={2020}
}

@Article{Wang:20,
	title={On solving minimax optimization locally: A follow-the-ridge approach.},
	author={Yuanhao, Wang and Guodong, Zhang and Jimmy, Ba},
	journal={Proc. - Int. Conf. Learn. Represent. (ICLR)},
	volume={},
	year={2020},
	month={Apr},
}

@Article{Maher:19,
	title={Solving a class of non-convex min-max games using iterative first order methods},
	author={Nouiehed, Maher and Sanjabi, Maziar and Huang, Tianjian and Lee, Jason D and Razaviyayn, Meisam},
	journal={Proc. Adv. Neural Inf. Process. Syst. (NeurIPS)},
	volume={32},
	year={2019}
}

@ARTICLE{Lei:25,
	author={Li, Lei and Zhang, Jiawei and Chang, Tsung-Hui},
	journal={IEEE J. Sel. Areas Commun.}, 
	title={Beamforming Optimization for Robust Sensing and Communication in Dynamic mmWave MIMO Networks}, 
	year={2025},
	month={Jan},
	volume={43},
	number={4},
	pages={1354-1370},
	doi={10.1109/JSAC.2025.3531545}}

@Article{Rui:24,
	author={Zhao, Rui and Hu, Xiaoyan and Liu, Chaowen and Wang, Wenjie and Liu, Boyang and Wong, Kai-Kit},
	journal={Proc. IEEE Chin. Int. Conf. Commun. (ICCC)}, 
	title={Joint Transmit and Receive Beamforming Design for Secure Communications in ISAC Systems with Eavesdropper and Jammer}, 
	year={2024},
	month={Aug},
	volume={},
	number={},
	pages={885-890},
	doi={10.1109/ICCC62479.2024.10681811}}

@ARTICLE{Dongfang:22,
	author={Xu, Dongfang and Yu, Xianghao and Ng, Derrick Wing Kwan and Schmeink, Anke and Schober, Robert},
	journal={IEEE Trans. Commun.}, 
	title={Robust and Secure Resource Allocation for ISAC Systems: A Novel Optimization Framework for Variable-Length Snapshots}, 
	year={2022},
	month={Nov},
	volume={70},
	number={12},
	pages={8196-8214},
	doi={10.1109/TCOMM.2022.3218629}}

@ARTICLE{Liu:24,
	author={Liu, Zhenrong and Li, Zongze and Gong, Yi and Wu, Yik-Chung},
	journal={IEEE Trans. Wireless Commun.}, 
	title={{RIS}-Aided Cooperative Mobile Edge Computing: Computation Efficiency Maximization via Joint Uplink and Downlink Resource Allocation}, 
	year={2024},
	month={Apr},
	volume={23},
	number={9},
	pages={11535-11550},
	doi={10.1109/TWC.2024.3382759}}

@ARTICLE{Sheng:18,
	author={Sheng, Zhichao and Tuan, Hoang Duong and Nasir, Ali Arshad and Duong, Trung Q. and Poor, H. Vincent},
	journal={IEEE Trans. Wireless Commun.}, 
	title={Power Allocation for Energy Efficiency and Secrecy of Wireless Interference Networks}, 
	year={2018},
	month={Jun},
	volume={17},
	number={6},
	pages={3737-3751},
	doi={10.1109/TWC.2018.2815626}}

@ARTICLE{Wei-Chiang:15,
	author={Li, Wei-Chiang and Chang, Tsung-Hui and Chi, Chong-Yung},
	journal={IEEE Trans. Signal Process.}, 
	title={Multicell Coordinated Beamforming With Rate Outage Constraint—Part II: Efficient Approximation Algorithms}, 
	year={2015},
	month={Mar},
	volume={63},
	number={11},
	pages={2763-2778},
	doi={10.1109/TSP.2015.2414896}}

@ARTICLE{Jiancheng:23,
	author={An, Jiancheng and Li, Hongbin and Ng, Derrick Wing Kwan and Yuen, Chau},
	journal={IEEE Trans. Wireless Commun.}, 
	title={Fundamental Detection Probability vs. Achievable Rate Tradeoff in Integrated Sensing and Communication Systems}, 
	year={2023},
	month={May},
	volume={22},
	number={12},
	pages={9835-9853},
	doi={10.1109/TWC.2023.3273850}}

@ARTICLE{Zixiang:24,
	author={Ren, Zixiang and Xu, Jie and Qiu, Ling and Wing Kwan Ng, Derrick},
	journal={IEEE J. Sel. Areas Commun.}, 
	title={Secure Cell-Free Integrated Sensing and Communication in the Presence of Information and Sensing Eavesdroppers}, 
	year={2024},
	month={Nov},
	volume={42},
	number={11},
	pages={3217-3231},
	doi={10.1109/JSAC.2024.3431582}}

@ARTICLE{Yinhong:24,
	author={Liu, Yinhong and Jin, Ming and Guo, Qinghua and Yao, Junteng},
	journal={IEEE Commun. Lett.}, 
	title={Secure Beamforming for NOMA-ISAC With System Imperfections}, 
	year={2024},
	month={May},
	volume={28},
	number={7},
	pages={1559-1563},
	doi={10.1109/LCOMM.2024.3398777}}

@ARTICLE{Peng:22,
	author={Liu, Peng and Fei, Zesong and Wang, Xinyi and Li, Bin and Huang, Yuzhen and Zhang, Zhi},
	journal={IEEE Wireless Commun.}, 
	title={Outage Constrained Robust Secure Beamforming in Integrated Sensing and Communication Systems}, 
	year={2022},
	month={Aug},
	volume={11},
	number={11},
	pages={2260-2264},
	doi={10.1109/LWC.2022.3198683}}

@ARTICLE{Siyi:25,
	author={Li, Siyi and Dong, Heng and Shan, Chengzhao and Fang, Xiaojie and Wu, Wei and Li, Zhuoming},
	journal={IEEE Trans. Veh. Technol.}, 
	title={Secure Hybrid Beamforming Design for Mmwave Integrated Sensing and Communication Systems}, 
	year={2025},
	month={Feb},
	volume={},
	number={},
	pages={1-16},
	doi={10.1109/TVT.2025.3544397}}

@article{Jiefei:24,
	title={An approximation proximal gradient algorithm for nonconvex-linear minimax problems with nonconvex nonsmooth terms},
	author={He, Jiefei and Zhang, Huiling and Xu, Zi},
	journal={J. Glob. Optim.},
	volume={90},
	number={1},
	pages={73--92},
	year={2024},
	month={Mar},
	publisher={Springer}
}

@ARTICLE{Songtao:20,
	author={Lu, Songtao and Tsaknakis, Ioannis and Hong, Mingyi and Chen, Yongxin},
	journal={IEEE Trans. Signal Process.}, 
	title={Hybrid Block Successive Approximation for One-Sided Non-Convex Min-Max Problems: Algorithms and Applications}, 
	year={2020},
	month={Apr},
	volume={68},
	number={},
	pages={3676-3691},
	doi={10.1109/TSP.2020.2986363}}

@book{Bertrand:23,
	title={The summation of series},
	author={Davis, Harold T},
	year={2014},
	publisher={Courier Dover Publications}
}

@ARTICLE{Yang:20,
	author={Lu, Yang and Xiong, Ke and Fan, Pingyi and Ding, Zhiguo and Zhong, Zhangdui and Letaief, Khaled Ben},
	journal={IEEE Trans. Wireless Commun.}, 
	title={Secrecy Energy Efficiency in Multi-Antenna SWIPT Networks With Dual-Layer PS Receivers}, 
	year={2020},
	month={Apr},
	volume={19},
	number={6},
	pages={4290-4306},
	doi={10.1109/TWC.2020.2982383}}

@ARTICLE{Boxiang:24,
	author={He, Boxiang and Wang, Fanggang and Cheng, Julian},
	journal={IEEE Trans. Wireless Commun.}, 
	title={Joint Secure Transceiver Design for Integrated Sensing and Communication}, 
	year={2024},
	month={Oct},
	volume={23},
	number={10},
	pages={13377-13393},
	doi={10.1109/TWC.2024.3400849}}

@ARTICLE{Zhenyao:23,
	author={He, Zhenyao and Xu, Wei and Shen, Hong and Ng, Derrick Wing Kwan and Eldar, Yonina C. and You, Xiaohu},
	journal={IEEE J. Sel. Areas Commun.}, 
	title={Full-Duplex Communication for ISAC: Joint Beamforming and Power Optimization}, 
	year={2023},
	month={Jun},
	volume={41},
	number={9},
	pages={2920-2936},
	doi={10.1109/JSAC.2023.3287540}}

@ARTICLE{Ziang:23,
	author={Liu, Ziang and Aditya, Sundar and Li, Hongyu and Clerckx, Bruno},
	journal={IEEE J. Sel. Areas Commun.}, 
	title={Joint Transmit and Receive Beamforming Design in Full-Duplex Integrated Sensing and Communications}, 
	year={2023},
	month={Sep},
	volume={41},
	number={9},
	pages={2907-2919},
	doi={10.1109/JSAC.2023.3287542}}

@ARTICLE{Suzhi:19,
	author={Bi, Suzhi and Lyu, Jiangbin and Ding, Zhi and Zhang, Rui},
	journal={IEEE Wireless Commun.}, 
	title={Engineering Radio Maps for Wireless Resource Management}, 
	year={2019},
	month={Feb},
	volume={26},
	number={2},
	pages={133-141},
	doi={10.1109/MWC.2019.1800146}}

@ARTICLE{Yang:25,
	author={Cao, Yang and Duan, Lingjie and Zhang, Rui},
	journal={IEEE Trans. Wireless Commun.}, 
	title={Sensing for Secure Communication in ISAC: Protocol Design and Beamforming Optimization}, 
	year={2025},
	month={Feb},
	volume={24},
	number={2},
	pages={1207-1220},
	doi={10.1109/TWC.2024.3506628}}

\end{document}